\documentclass{llncs}
\usepackage{amssymb,amsmath,stmaryrd,wasysym}
\usepackage{gastex}
\usepackage{multirow}
\usepackage{epsfig}
\usepackage{arydshln}

\newcounter{compressEnum}

{}

\newtheorem{mydefi}{Definition}

\newcommand{\ee}{\hfill \ensuremath{\blacksquare}}  

\newcommand{\Pref}{{\sf Prefs}}
\newcommand{\Play}{{\sf Plays}}
\newcommand{\K}{{\sf K}}
\newcommand{\Last}{{\sf Last}}

\newcommand{\Mem}{{\sf Mem}}

\newcommand{\Num}{{\sf Num}}

\def\abs#1{\ensuremath{\lvert #1\rvert}}

\newcommand{\Good}{\mathsf{Good}}

\newcommand{\Safe}{\mathsf{Safe}}
\newcommand{\Win}{\mathsf{Win}}

\newcommand{\trans}{{\delta}}

\newcommand{\real}{{\mathbb R}}

\renewcommand{\l}{{\ell}}

\newcommand{\nat}{\mathbb N} 
\newcommand{\tuple}[1]{\langle #1 \rangle}

\newcommand{\Post}{\mathsf{Post}}

\newcommand{\Succ}{\mathsf{Succ}}
\newcommand{\Abs}{\mathsf{Abs}}

\newcommand{\Obs}{{\cal{O}}}

\newcommand{\obs}{\mathsf{obs}}

\newcommand{\target}{{\cal T}}

\newcommand{\sink}{{\sf sink}}

\newcommand{\straa}{\sigma}
\newcommand{\strab}{\pi}
\newcommand{\dist}{{\cal D}}
\newcommand{\Supp}{{\sf Supp}}

\newcommand{\Straa}{\Sigma}
\newcommand{\Strab}{\Pi}

\newcommand{\rank}{\mathsf{Rank}}
\newcommand{\maxrank}{\mathsf{MaxRank}}

\newcommand{\Reach}{\mathsf{Reach}}

\newcommand{\Buchi}{{\sf B\ddot{u}chi}}

\newcommand{\Prb}{\mathrm{Pr}}
\newcommand{\set}[1]{\{\: #1 \:\}}

\newcommand{\Nats}{\mathbb{N}}
\newcommand{\slopefrac}[2]{\leavevmode\kern.1em
  \raise .5ex\hbox{\the\scriptfont0 #1}\kern-.1em
  /\kern-.15em\lower .25ex\hbox{\the\scriptfont0 #2}}

\newcommand{\ov}{\overline}

\newcommand{\inc}{{\sf inc}}
\newcommand{\idle}{{\sf idle}}

\newcommand{\calw}{{\mathcal W}}
\newcommand{\calb}{{\mathcal B}}
\newcommand{\wh}{\widehat}
\newcommand{\PosReachSafe}{{\sc PosReachSureSafe}}

\def\mynote#1{{\sf $\clubsuit$ #1 $\clubsuit$}}

\makeatletter
\DeclareRobustCommand\sfrac[1]{\@ifnextchar/{\@sfrac{#1}}%
                                            {\@sfrac{#1}/}}
\def\@sfrac#1/#2{\leavevmode\scalebox{.9}{\kern.1em\raise.5ex
         \hbox{$\m@th\mbox{\fontsize\sf@size\z@
                           \selectfont#1}$}\kern-.1em
         /\kern-.15em\lower.25ex
          \hbox{$\m@th\mbox{\fontsize\sf@size\z@
                            \selectfont#2}$}}}
\DeclareRobustCommand\numfrac[1]{\@ifnextchar/{\@numfrac{#1}}%
                                            {\@numfrac{#1}}}
\def\@numfrac#1{\leavevmode \hbox{$\m@th\mbox{\fontsize\sf@size\z@
                           \selectfont#1}$}}
\makeatother

\makeatletter

\begingroup \catcode `|=0 \catcode `[= 1
\catcode`]=2 \catcode `\{=12 \catcode `\}=12
\catcode`\\=12 |gdef|@xcomment#1\end{comment}[|end[comment]]
|endgroup

\def\@comment{\let\do\@makeother \dospecials\catcode`\^^M=10\def\par{}}

\def\begincomment{\@comment\@xcomment}

\makeatother

\newenvironment{comment}{\begincomment}{}

\begin{document}
\pagestyle{plain}

\title{{\bf Partial-Observation Stochastic Games: \\ How to Win when Belief Fails}
}

\author{Krishnendu Chatterjee\inst{1} \and Laurent Doyen\inst{2}}

\institute{
IST Austria (Institute of Science and Technology Austria) \\
\and LSV, ENS Cachan \& CNRS, France \\
}

\maketitle
\pagestyle{plain}

\begin{abstract}
In two-player finite-state stochastic games of partial observation on graphs,  
in every state of the graph, the players simultaneously choose an action, and 
their joint actions determine a probability distribution over the successor 
states. 
The game is played for infinitely many rounds and thus the players construct an 
infinite path in the graph. We consider reachability objectives where the first player tries to 
ensure a target state to be visited almost-surely (i.e., with probability~$1$)
or positively (i.e., with positive probability), no matter the strategy of the second
player.

We classify such games according to the information and to the power 
of randomization available to the players. On the basis of information,
the game can be \emph{one-sided} with either $(a)$ player~$1$, or $(b)$ player~$2$ 
having partial observation (and the other player has perfect observation), 
or \emph{two-sided} with $(c)$ both players having partial observation.
On the basis of randomization, $(a)$ the players may not be allowed to use randomization (pure strategies),
or $(b)$ they may choose a probability distribution over actions but the actual random 
choice is external and not visible to the player (actions invisible), or $(c)$ 
they may use full randomization. 

Our main results for pure strategies are as follows:
(1) For one-sided games with player~2 perfect observation we
show that (in contrast to full randomized strategies) \emph{belief-based} 
(subset-construction based) strategies are not sufficient, and we present an 
exponential upper bound on memory both for almost-sure and positive winning strategies; 
we show that the problem of deciding the existence of almost-sure and positive winning strategies 
for player~1 is EXPTIME-complete and present symbolic algorithms that avoid the explicit
exponential construction.
(2) For one-sided games with player~1 perfect observation 
we show that non-elementary memory is both necessary and sufficient for both almost-sure
and positive winning strategies.
(3) We show that for the general (two-sided) case finite-memory strategies are sufficient 
for both positive and almost-sure winning, and at least non-elementary memory is required.
We establish the equivalence of the almost-sure winning problems for pure strategies and for   
randomized strategies with actions invisible.
Our equivalence result exhibit serious flaws in previous results in the literature: we 
show a non-elementary memory lower bound for almost-sure winning whereas an exponential
upper bound was previously claimed. 
\end{abstract}

\section{Introduction}

\smallskip\noindent{\bf Games on graphs.}
Two-player games on graphs play a central role in several important 
problems in computer science, such as controller synthesis~\cite{PR89,RamadgeWonham87}, 
verification of open systems~\cite{AHK02}, realizability and compatibility 
checking~\cite{AbadiLamportWolper,Dill89book,InterfaceAutomata}, and many others.
Most results about two-player games on graphs make the
hypothesis of {\em perfect observation} (i.e., both players have perfect or complete
observation about the state of the game).  
This assumption is often not realistic in practice.
For example in the context of hybrid systems, the controller 
acquires information about the state of a plant using digital sensors 
with finite precision, which gives imperfect information about the state of the plant~\cite{DDR06,HK99}.
Similarly, in a concurrent system where the players represent individual processes, 
each process has only access to the public variables of the other processes, not to 
their private variables~\cite{Reif84,AHK02}.
Such problems are better modeled in the more general framework
of \emph{partial-observation} games~\cite{Reif79,Reif84,RP80,CDHR07,BGG09} and have been 
studied in the context of verification and synthesis~\cite{KV00a,DF08}
(also see~\cite{AMV07} for pushdown partial-observation games). 

\smallskip\noindent{\bf Partial-observation stochastic games and subclasses.}
In two-player partial-observation stochastic games on graphs with a finite
state space, in every round, both players independently and simultaneously 
choose actions which along with the current state give a probability distribution 
over the successor states in the game.
In a general setting, the players may not be able to distinguish certain states
which are observationally equivalent for them (e.g., if they differ only by the value of 
private variables). The state space is partitioned into \emph{observations} 
defined as equivalence classes and the players do not see the actual state of 
the game, but only an observation (which is typically different for the two players).
The model of partial-observation games we consider is the same as the model 
of stochastic games with signals~\cite{BGG09} and is a standard model in 
game theory~\cite{RSV03,SorinBook}. It subsumes other classical game models
such as concurrent games~\cite{Sha53,AHK07}, probabilistic automata~\cite{RabinProb63,Bukharaev,PAZBook}, 
and partial-observation Markov decision processes (POMDPs)~\cite{PT87} 
(see also the recent decidability and complexity results for probabilistic 
automata~\cite{BBG08,BBG09,BG05,CSV09a,CSV09,CSV10,GO10} and for POMDPs~\cite{CDH10a,BBG08,TBG09}).

The special case of \emph{perfect observation} for a player corresponds to 
every observation for this player being a singleton. 
%
%
Depending on which player has perfect observation, 
we consider the following \emph{one-sided} subclasses of the general two-sided 
partial-observation stochastic games: 
(1)~\emph{player~$1$ partial and player~$2$ perfect} where player~2 has perfect observation, and 
player~1 has partial observation; and
(2)~\emph{player~$1$ perfect and player~$2$ partial} where player~1 has perfect observation, and
player~2 has partial observation.
The case where the two players have perfect observation corresponds to the
well-known perfect-information (perfect-observation) stochastic 
games~\cite{Sha53,Con92,AHK07}.

Note that in a given game~$G$, if player~$1$ wins in the setting of 
player~$1$ partial and player~$2$ perfect, then player~$1$ wins in the 
game $G$ as well. Analogously, if player~$1$ cannot win in the setting of 
player~$1$ perfect and player~$2$ partial, then player~$1$ does not win 
in the game $G$ either. In this sense, the one-sided games are conservative
over- and under-approximations of two-sided games. 
In the context of applications in verification and synthesis, the conservative 
approximation is that the adversary is all powerful, and hence player~1 partial
and player~2 perfect games provide the important worst-case analysis of partial-observation 
games. 


\smallskip\noindent{\bf Objectives and qualitative problems.}
In this work we consider partial-observation stochastic games with 
\emph{reachability} objectives where the goal of player~1 is to reach a 
set of target states and the goal of player~2 is to prevent player~1 from 
reaching the target states.
The study of partial-observation games is considerably more complicated than 
games of perfect observation. 
For example, in contrast to perfect-observation games, 
strategies in partial-observation games require both randomization and
memory for reachability objectives; and the \emph{quantitative} problem of deciding whether there exists a  
strategy for player~1 to ensure that the target is reached with probability 
at least $\frac{1}{2}$ can be decided in NP $\cap$ coNP for perfect-observation 
stochastic games~\cite{Con92}, whereas the problem is undecidable even for 
partial-observation stochastic games with only one player~\cite{PAZBook}.
Since the quantitative problem is undecidable we consider the following 
\emph{qualitative} problems: the \emph{almost-sure} (resp. \emph{positive}) problem asks
whether there exists a strategy for player~1 to ensure that the target set is 
reached with probability~1 (resp. positive probability).

\smallskip\noindent{\bf Classes of strategies.}
In general, randomized strategies are necessary to win with probability~$1$ 
in a partial-observation game with reachability objective~\cite{CDHR07}. 
However, there exist two types of randomized strategies where either $(i)$ 
actions are visible, the player can observe the action he played~\cite{CDHR07,BGG09},
or $(ii)$ actions are invisible, the player may choose a probability distribution
over actions, but the source of randomization is external and the
actual choice of the action is invisible to the player~\cite{GS09}.
The second model is more general since the qualitative problems of randomized 
strategies with actions visible can be reduced in polynomial time to randomized strategies 
with actions invisible, by modeling the visibility of actions using the observations on states.

With actions visible, the almost-sure (resp. positive) problem was shown to be EXPTIME-complete (resp. PTIME-complete) for
one-sided games with player~1 partial and player~2 perfect~\cite{CDHR07}, and 2EXPTIME-complete (resp. EXPTIME-complete) 
in the two-sided case~\cite{BGG09}.
For the positive problem memoryless randomized strategies exist, and 
for the  almost-sure problem \emph{belief-based} strategies exist 
(strategies based on subset construction that consider the possible 
current states of the game). 

It was remarked (without any proof) in~\cite[p.4]{CDHR07} that these results easily 
extend to randomized strategies with actions invisible for one-sided games with player~1
partial and player~2 perfect. 
It was claimed in~\cite{GS09} (Theorems~1 \& 2) that the almost-sure problem 
is 2EXPTIME-complete for randomized strategies with actions invisible for two-sided games, 
and that belief-based strategies are sufficient for player~$1$.
Thus it is believed that the two qualitative problems with actions visible or 
actions invisible  are essentially equivalent.

In this paper, we consider the class of \emph{pure} strategies, which do not
use randomization at all. Pure strategies arise naturally in the implementation of 
controllers and processes that do not have access to any source of randomization.
Moreover we will establish deep connections between the qualitative problems for pure strategies 
and for randomized strategies with actions invisible, which on one hand 
exhibit major flaws in previous results of the literature (the remark without proof 
of~\cite{CDHR07} and the main results of~\cite{GS09}), and on the other hand show that 
the solution for almost-sure winning randomized strategies with actions invisible (which is the most
general case) can be surprisingly obtained by solving the problem for pure strategies.


\smallskip\noindent{\bf Contributions.} The contributions of the paper are summarized below.
\begin{enumerate}
\item \emph{Player~$1$ partial and player~$2$ perfect.} 
We show that both for almost-sure and positive winning, belief-based 
pure strategies are not sufficient. This implies that the classical 
approaches relying on the belief-based subset construction cannot work
for solving the qualitative problems for pure strategies. However, we present an optimal
exponential upper bound on the memory needed by pure strategies 
(the exponential lower bound follows from the special case of non-stochastic games~\cite{BD08}). 
By a reduction to a perfect-observation game of exponential size, we show 
that both the almost-sure and positive problems are EXPTIME-complete 
for one-sided games with perfect-observation for player~$2$.
In contrast to the previous proofs of EXPTIME upper bound that rely
either on subset constructions 
or enumeration of belief-based strategies, 
our correctness proof relies on a novel rank-based 
argument that works uniformly both for positive and almost-sure winning.
The structure of this construction also provides symbolic antichain-based
algorithms (see~\cite{DR10} for a survey of the antichain approach) for solving the 
qualitative problems that avoids the explicit exponential construction.
Thus for the important special case of player~1 partial and player~2 perfect
we establish optimal memory bound, complexity bound, and present 
symbolic algorithmic solutions for the qualitative problems. 

\item \emph{Player~$1$ perfect and player~$2$ partial.} 

\begin{enumerate}
\item We show a very surprising result that both for positive and almost-sure winning,
pure strategies for player~$1$ require memory of non-elementary size (i.e., a tower of
exponentials). This is in sharp contrast with $(i)$ the case of randomized strategies 
(with or without actions visible) where memoryless strategies are sufficient 
for positive winning, and with $(ii)$ the previous case where player~1 has partial observation 
and player~2 has perfect observation, where pure strategies for positive winning require only 
exponential memory.
Surprisingly and perhaps counter-intuitively when player~1 has more information
and player~2 has less information, the positive winning strategies for player~1 
require much more memory (non-elementary as compared to exponential).
With more information player~1 can win from more states, but the winning 
strategy is much harder to implement.

\item We present a non-elementary upper bound for the memory needed by 
pure strategies for positive winning.
We then show with an example that for almost-sure winning more memory may be required
as compared to positive winning.
Finally, we show how to combine pure strategies for positive winning in 
a recharging scheme to obtain a non-elementary upper bound for the memory required
by pure strategies for almost-sure winning. 
Thus we establish non-elementary complete bounds for pure strategies
both for positive and almost-sure winning.
\end{enumerate}

\item \emph{General (two-sided) case.} 
We show that in the general case finite memory strategies are 
sufficient both for positive and almost-sure winning. 
The result is obtained essentially 
by a simple generalization of K\"onig's Lemma~\cite{Konig36}.
The non-elementary lower bound for memory follows from the special case when 
player~1 has perfect observation and player~2 has partial observation.

\item \emph{Randomized strategies with actions invisible.}
For randomized strategies with actions invisible we present two reductions
to establish connections with pure strategies. 
First, we show that the almost-sure problem for randomized strategies
with actions invisible can be reduced in polynomial time to the almost-sure problem
for pure strategies. The reduction requires to first establish that finite-memory 
randomized strategies are sufficient in two-sided games.
Second, we show that the problem of almost-sure winning with pure strategies 
can be reduced in polynomial time to the problem of randomized strategies with 
actions invisible. For this reduction it is crucial that the actions are not
visible. 

Our reductions have deep consequences. They unexpectedly imply that the problems
of almost-sure winning with \emph{pure} strategies or \emph{randomized} strategies with actions
invisible are polynomial-time \emph{equivalent}. Moreover, it follows that 
even in one-sided games with player~1 partial and player~2 perfect, 
belief-based randomized strategies (with actions invisible) 
are not sufficient for almost-sure winning.
This shows that the remark (without proof) of~\cite{CDHR07} that the results (such as
existence of belief-based strategies) of randomized strategies with actions visible 
carry over to actions invisible is an oversight.
However from our first reduction and our results for pure strategies it follows
that there is an exponential upper bound on memory and the problem is 
EXPTIME-complete for one-sided games with player~1 partial and player~2 perfect.
More importantly, our results exhibit a serious flaw in the main result of~\cite{GS09}
which showed that belief-based randomized strategies with 
actions invisible are sufficient for almost-sure winning in two-sided games, and concluded  
that enumerating over such strategies yields a 2EXPTIME algorithm for the problem. 
Our second reduction and lower bound for pure strategies show that the result is incorrect,
and that the exponential (belief-based) upper bound is far off.
Instead, 
the lower bound on memory for almost-sure winning 
with randomized strategies and actions invisible is non-elementary. 
Thus, contrary to the general belief, there is a sharp contrast for randomized 
strategies with or without actions visible: if actions are visible, then exponential 
memory is sufficient for almost-sure winning while if actions are not visible, 
then memory of non-elementary size is necessary in general.
\end{enumerate}

The memory requirements are summarized in Table~\ref{tab:memoryreq} 
and the results of this paper are shown in bold font. 
We explain how the other results of the table follow from results of the literature. 
For randomized strategies (with or without actions visible), if a positive winning 
strategy exists, then a memoryless strategy that plays all actions uniformly at random
is also positive winning. 
Thus the memoryless result for positive winning strategies follows for all cases of
randomized strategies.
The belief-based bound for memory of almost-sure winning randomized strategies
with actions visible follows from~\cite{CDHR07,BGG09}. 
The memoryless strategies results for almost-sure winning for one-sided games 
with player~1 perfect and player~2 partial are obtained as follows: 
when actions are visible, then belief-based strategies coincide
with memoryless strategies as player~1 has perfect observation.
If player~1 has perfect observation, then for memoryless strategies whether actions are visible 
or not is irrelevant and thus the memoryless result also follows for randomized strategies
with actions invisible. 
Thus along with our results we obtain Table~\ref{tab:memoryreq}.



\begin{table}[h]
\begin{center}
\begin{scriptsize}
\begin{tabular}{|l|c|c|c|c|c|c|}
\cline{2-7}

\multicolumn{1}{l}{}              & \multicolumn{2}{|c|}{one-sided}          & \multicolumn{2}{|c|}{one-sided}          & \multicolumn{2}{|c|}{\multirow{2}{*}{two-sided}} \\
\multicolumn{1}{l}{}              & \multicolumn{2}{|c|}{player~$2$ perfect} & \multicolumn{2}{|c|}{player~$1$ perfect} & \multicolumn{2}{|c|}{}   \\
\cline{2-7}
\multicolumn{1}{l|}{{\small \strut}} & Positive     & Almost-sure               &  Positive     & Almost-sure              & Positive     & Almost-sure     \\
\hline
Randomized {\small \strut}           & Memoryless   & Exponential               &  Memoryless   & Memoryless               & Memoryless   & Exponential     \\
(actions visible) {\small \strut}    &              & (belief-based)            &               &                          &              & (belief-based)  \\
\hline
Randomized {\small \strut}           & Memoryless   & {\bf Exponential}         &  Memoryless   & Memoryless               & Memoryless   & {\bf Non-elem.} \\
(actions invisible) {\small \strut}  &              & {\bf (more than }         &               &                          &              & {\bf low. bound}  \\
 \cdashline{7-7}                                    &              & {\bf belief)}             &               &                          &              & {\bf Finite }      \\
                                     &              &                           &               &                          &              & {\bf upp. bound}      \\
\hline
Pure {\small \strut}                 & {\bf Exponential} & {\bf Exponential}    & {\bf Non-elem.} & {\bf Non-elem.}        & {\bf Non-elem.}  & {\bf Non-elem.} \\
                                     & {\bf (more than } & {\bf (more than}     & {\bf complete}  & {\bf complete}         & {\bf low. bound} & {\bf low. bound}  \\
\cdashline{6-7}
                                     & {\bf belief)}     & {\bf belief)}        &                 &                        & {\bf Finite } &  {\bf Finite }  \\
                                     &                   &                      &                 &                        & {\bf upp. bound} &  {\bf upp. bound}  \\
\hline
\end{tabular}
\end{scriptsize}
\end{center}
\caption{Memory requirement for player~$1$ and reachability objective.}\label{tab:memoryreq}
\end{table}


\section{Definitions}\label{sec:definitions}


A \emph{probability distribution} on a finite set $S$ is a function
$\kappa: S \to [0,1]$ such that $\sum_{s \in S} \kappa(s) = 1$. 
The \emph{support} of $\kappa$ is the set $\Supp(\kappa) = \{s \in S \mid \kappa(s) > 0\}$.
We denote by $\dist(S)$ the set of probability distributions on $S$.
Given $s \in S$, the \emph{Dirac distribution} on $s$ assigns probability~$1$
to $s$.

\smallskip\noindent{\em Games.}
Given finite alphabets $A_i$ of actions for player~$i$ ($i=1,2$),
a \emph{stochastic game} on $A_1, A_2$ is a tuple $G=\tuple{Q, q_0, \trans}$     
where $Q$ is a finite set of states, 
$q_0 \in Q$ is the initial state, and
$\trans: Q \times A_1 \times A_2 \to \dist(Q)$ is a 
probabilistic transition function that, given a current state $q$ and
actions~$a,b$ for the players gives the transition probability 
$\trans(q,a,b)(q')$ to the next state~$q'$. The game is called \emph{deterministic}
if $\trans(q,a,b)$ is a  Dirac distribution for all $(q,a,b) \in Q\times A_1 \times A_2$.
A state $q$ is \emph{absorbing} if $\trans(q,a,b)$ is the Dirac distribution on $q$
for all $(a,b) \in A_1 \times A_2$.
In some examples, we allow an initial distribution of states. This can be encoded
in our game model by a probabilistic transition from the initial state.

A \emph{player-$1$ state} is a state $q$ where $\trans(q,a,b) = \trans(q,a,b')$
for all $a \in A_1$ and all $b, b' \in A_2$. We use the notation $\trans(q,a,-)$.
\emph{Player-$2$ states} are defined analogously. In figures, we use boxes to emphasize 
that a state is a player-$2$ state, and we represent probabilistic branches
using diamonds (which are not real `states', e.g., as in \figurename~\ref{fig:GS09-wrong-positive}).

In a (two-sided) \emph{partial-observation}
game, the players have a partial or incomplete view of the states visited and 
of the actions played in the game. This view may be different for the two players 
and it is defined by equivalence relations $\approx_i$ on the states and on the actions.
For player~$i$, equivalent states (or actions) are indistinguishable. 
We denote by $\Obs_i \subseteq 2^Q$ ($i=1,2$) the equivalence classes of $\approx_i$
which define two partitions of the state space $Q$, and we call them \emph{observations} (for player~$i$).
These partitions uniquely define functions $\obs_i: Q \to \Obs_i$ 
($i=1,2$) such that $q \in \obs_i(q)$ for all $q \in Q$,
that map each state $q$ to its observation for player~$i$. 

In the case where all states and actions are equivalent (i.e., the relation 
$\approx_i$ is the set $(Q \times Q) \cup (A_1 \times A_1) \cup (A_2 \times A_2)$),
we say that player~$i$ is \emph{blind} and the actions are \emph{invisible}. 
In this case, we have $\Obs_i = \{Q\}$ because all states have the same observation.
Note that the case of perfect observation for player~$i$ corresponds to the case 
$\Obs_i = \{\{q_0\},\{q_1\}, \dots, \{q_n\}\}$ (given  $Q = \{q_0, q_1, \dots, q_n\}$),
and $a \approx_i b$ iff $a=b$, for all actions $a,b$.

For $s \subseteq Q$, $a \in A_1$, and $b \in A_2$, let  
$\Post_{a,b}(s) = \bigcup_{q \in s} \Supp(\trans(q,a,b))$ denote the set of possible successors
of $q$ given action $a$ and $b$,
and let $\Post_{a,-}(s) = \bigcup_{b \in A_2} \Post_{a,b}(s)$.

\medskip\noindent{\em Plays and observations.}
Initially, the game starts in the initial state $q_0$. 
In each round, player~$1$ chooses an action $a \in A_1$, 
player~$2$ (simultaneously and independently) chooses an action  $b \in A_2$, and the successor of the current
state~$q$ is chosen according to the probabilistic transition function $\trans(q,a,b)$.
A \emph{play} in $G$ is an infinite sequence 
$\rho=q_0 a_0 b_0 q_1 a_1 b_1 q_2 \ldots$ such that $q_0$ is the initial state and 
$\trans(q_j,a_j,b_j)(q_{j+1}) > 0$ for all $j \geq 0$
(the actions $a_j$'s and $b_j$'s are the actions \emph{associated} to the play). 
Its \emph{length} is $\abs{\rho} = \infty$. 
The length of a play prefix $\rho=q_0 a_0 b_0 q_1 \ldots q_k$ is 
$\abs{\rho} = k$, and its last element is $\Last(\rho) = q_k$. 
A state $q \in Q$ is \emph{reachable} if it occurs in some play.
We denote by $\Play(G)$ the set of plays in $G$, 
and by $\Pref(G)$ the set of corresponding finite prefixes.
The \emph{observation sequence} for player~$i$ ($i=1,2$) 
of a play (prefix) $\rho$  is the unique (in)finite sequence 
$\obs_i(\rho) = \gamma_0 \gamma_1 \ldots$ 
such that $q_j \in \gamma_j \in \Obs_i$  
for all $0  \leq j \leq \abs{\rho}$. 

The games with \emph{one-sided partial-observation} are the special case 
where either $\approx_1$ is equality and hence $\Obs_1 = \{ \{q\} \mid q \in Q \}$ (player~1 has 
complete observation) or $\approx_2$ is equality and hence $\Obs_2 = \{ \{q\} \mid q \in Q \}$ (player~2 has 
complete observation). The games with \emph{perfect observation} are the special 
cases where $\approx_1$ and $\approx_2$ are equality, i.e., every state and action 
is visible to both players.

\smallskip\noindent{\em Strategies.}
A \emph{pure strategy} in $G$ for player~$1$ is a function $\straa:\Pref(G) \to A_1$. 
A \emph{randomized strategy} in $G$ for player~$1$ is a function $\straa:\Pref(G) \to \dist(A_1)$. 
A (pure or randomized) strategy $\straa$ for player~$1$ is 
\emph{observation-based} if for all prefixes $\rho = q_0 a_0 b_0 q_1 \ldots$ and
$\rho' = q'_0 a'_0 b'_0 q'_1 \ldots$, 
if $a_j \approx_1 a'_j$ and $b_j \approx_1 b'_j$ for all $j \geq 0$, and
$\obs_1(\rho)=\obs_1(\rho')$, then $\straa(\rho)=\straa(\rho')$. 
It is assumed that strategies are observation-based in partial-observation games.
If for all actions $a$ and $b$ we have $a \approx_1 b$ and $a \approx_2 b$ iff $a=b$ 
(all actions are distinguishable), 
then the strategy is \emph{action visible}, and if for all actions $a$ and $b$ 
we have $a \approx_1 b$ and $a \approx_2 b$ (all actions are indistinguishable), 
then the strategy is \emph{action invisible}.
We say that a play (prefix) $\rho=q_0 a_0 b_0 q_1 \ldots$ is \emph{compatible} 
with a pure (resp., randomized) strategy~$\straa$ 
if the associated action of player~$1$ in step $j$ is $a_j = \straa(q_0 a_0 b_0 \ldots q_{j-1})$
(resp., $a_j \in \Supp(\straa(q_0 a_0 b_0 \ldots q_{j-1}))$)
for all $0  \leq j \leq \abs{\rho}$.

We omit analogous definitions of strategies for player~$2$.
We denote by $\Straa_G$, $\Straa_G^O$, $\Straa_G^P$, $\Strab_G$, $\Strab_G^O$, and $\Strab_G^P$ 
the set of all player-$1$ strategies, the set of all observation-based player-$1$ strategies, 
the set of all pure player-$1$ strategies, the set of all player-$2$ strategies in $G$, the set of all 
observation-based player-$2$ strategies, and the set of all pure player-$2$ strategies, respectively.

\paragraph{Remarks.}
\begin{enumerate}
\item The model of games with partial observation on both actions and states
can be encoded in a model of games with actions invisible and observations 
on states only: when actions are invisible, we can use the state space to keep 
track of the last action played, and reveal information about the last action 
played using observations on the states.
Therefore, in the sequel we assume that the actions are invisible to the players
with partial observation. A play is then viewed as a sequence of states only,
and the definition of strategies is updated accordingly.
Note that a player with perfect observation has actions and states 
visible (and the equivalence relation $\approx_i$ is equality).



\item The important special case of partial-observation Markov decision processes (POMDP)
corresponds to the case where either all states in the game are player-$1$ states (player-1 POMDP)
or all states are player-$2$ states (player-2 POMDP). 
For POMDP it is known that randomization is not necessary, and 
pure strategies are as powerful as randomized strategies~\cite{CDGH10}.
%
\end{enumerate}

\smallskip\noindent{\em Finite-memory strategies.}
A player-1 strategy uses \emph{finite-memory} if it can be encoded
by a deterministic transducer $\tuple{\Mem, m_0, \alpha_u, \alpha_n}$
where $\Mem$ is a finite set (the memory of the strategy), 
$m_0 \in \Mem$ is the initial memory value,
$\alpha_u: \Mem \times \Obs_1  \to \Mem$ is an update function, and 
$\alpha_n: \Mem \times \Obs_1 \to \dist(A_1)$ is a next-move function. 
The \emph{size} of the strategy is the number $\abs{\Mem}$ of memory values.
If the current observation is $o$, and the current memory value is $m$,
then the strategy chooses the next action according to 
the probability distribution $\alpha_n(m,o)$,
and the memory is updated to $\alpha_u(m,o)$. 
Formally, $\tuple{\Mem, m_0, \alpha_u, \alpha_n}$
defines the strategy $\straa$ such that $\straa(\rho\cdot q) = \alpha_n(\hat{\alpha}_u(m_0, \obs_1(\rho)), \obs_1(q))$
for all $\rho \in Q^*$ and $q \in Q$, where $\hat{\alpha}_u$ extends $\alpha_u$ to sequences
of observations as expected. 
This definition extends to infinite-memory strategies by dropping the 
assumption that the set $\Mem$ is finite. 
A strategy is \emph{memoryless} if $\abs{\Mem} = 1$.
For a strategy~$\straa$, we denote by $G_{\straa}$ the player-2 POMDP 
obtained as the synchronous product of~$G$ with the transducer defining $\straa$.

\smallskip\noindent{\em Objectives and winning modes.} 
An \emph{objective} (for player~$1$) in $G$ is a set $\phi \subseteq \Play(G)$ of plays.
A play $\rho \in \Play(G)$ \emph{satisfies} the objective $\phi$, denoted $\rho \models \phi$, if 
$\rho \in \phi$.
Objectives are generally Borel measurable: 
a Borel objective is a Borel set in the Cantor topology~\cite{Kechris}.
Given strategies $\straa$ and $\strab$ for the two players, 
the probabilities of a measurable objective~$\phi$ is uniquely defined~\cite{Var85}. 
We denote by $\Prb_{q_0}^{\straa,\strab}(\phi)$ 
the probability that $\phi$ is satisfied by the play obtained from the starting 
state $q_0$ when the strategies $\straa$ and $\strab$ are used.

We specifically consider the following objectives. Given a set $\target \subseteq Q$ of target states,
the \emph{reachability objective} requires that the play visit the set $\target$:
$\Reach(\target)= \{q_0 a_0 b_0 q_1 \ldots \in \Play(G) \mid \exists i \geq 0: q_i \in \target\}$,
and the \emph{B\"uchi objective} requires that the play visit the set $\target$ infinitely often, 
$\Buchi(\target)= \{q_0 a_0 b_0 q_1 \ldots \in \Play(G) \mid \forall i \geq 0 \cdot \exists j \geq i: q_j \in \target\}$.
Our solution for reachability objectives will also use the dual notion of
\emph{safety objectives} that require the play to stay within the set $\target$: 
$\Safe(\target)= \{q_0 a_0 b_0 q_1 \ldots \in \Play(G) \mid \forall i \geq 0: q_i \in \target\}$.
In figures, the target states in $\target$ are double-lined and labeled by $\smiley$.

Given a game structure $G$ and a state $q$, an observation-based strategy $\straa$ for player~$1$ is 
\emph{almost-sure winning}  
(resp. \emph{positive winning}) for the objective $\phi$ from $q$ if 
for all observation-based randomized strategies $\strab$ for 
player~$2$, we have $\Prb_{q}^{\straa,\strab}(\phi)=1$ 
(resp.  $\Prb_{q}^{\straa,\strab}(\phi)>0$).
The strategy~$\straa$ is \emph{sure winning} if all plays compatible with $\straa$ satisfy $\phi$.
We also say that the state~$q$ is almost-sure (or positive, or sure) winning for player~$1$.



\smallskip\noindent{\em Positive and almost-sure winning problems.} 
We are interested in the problems of deciding, given a game structure $G$, a state $q$, and an objective $\phi$, 
whether there exists a $\{$pure, randomized$\}$ strategy
which is $\{$almost-sure, positive$\}$ winning from $q$ for the objective $\phi$.
For safety objectives almost-sure winning coincides with sure winning, 
however for reachability objectives they are different.
The sure winning problem for the objectives we consider has been studied 
in~\cite{Reif79,CDHR07,CD10b}.
The almost-sure winning problem for B\"uchi objectives can be easily reduced
to the almost-sure winning problem for reachability objectives~\cite{BBG08}, and 
the reduction is as follows: given a two-sided stochastic game with B\"uchi objective 
$\Buchi(\target)$, we add an absorbing state $q_T$, make $q_T$ the target 
state for the reachability objective, and from every state  $q \in \target$ we add 
positive probability transitions to $q_T$ (details and correctness proof follow
from~\cite[Lemma~13]{BBG08}).
The positive winning problem for B\"uchi objectives is undecidable even for 
POMDPs~\cite{BBG08}. 
Hence in this paper we only focus on reachability objectives.
In all our analysis, the counter strategies of player~2 can be restricted 
to pure strategies, because once a strategy for player~1 is fixed, then 
we obtain a POMDP for player~2 in which pure strategies are as powerful as
randomized strategies~\cite{CDGH10}.





\begin{figure}[!tb]
\hrule
\begin{center}
\def\fsize{\normalsize}

\begin{picture}(75,60)(0,0)

{\fsize


\node[Nmarks=i, Nmr=0](q0)(10,28.5){$q_0$}

\node[Nmarks=n](q1)(30,40){$q_1$}
\node[Nmarks=n](q2)(30,17){$q_2$}

\rpnode[Nmarks=n](r1)(50,40)(4,3.5){}
\rpnode[Nmarks=n](r2)(50,17)(4,3.5){}

\node[Nmarks=r](qX)(70,28.5){\smiley}

\drawedge[ELpos=45, ELside=l, ELdist=.5](q0,q1){$-,a$}
\drawedge[ELpos=45, ELside=r, ELdist=.5](q0,q2){$-,b$}

\drawloop[ELpos=50, ELside=l, ELdist=1, ELside=l,loopCW=y, loopdiam=6, loopangle=90](q1){$b,-$}
\drawloop[ELpos=50, ELside=l, ELdist=1, ELside=l,loopCW=y, loopdiam=6, loopangle=90](q2){$a,-$}
\drawedge[ELpos=50, ELside=l, ELdist=1, curvedepth=0](q1,r1){$a,-$}
\drawedge[ELpos=50, ELside=l, ELdist=1, curvedepth=0](q2,r2){$b,-$}

\drawedge[ELpos=18, ELside=l, ELdist=.5, curvedepth=3, sxo=2, syo=1.5, exo=-2](r1,qX){\sfrac{1}{2}}
\drawedge[ELpos=30, ELside=l, ELdist=.5](r2,qX){\sfrac{1}{2}}
\drawbpedge[ELpos=28, ELside=r, ELdist=1, eyo=-1](r1,330,20,q1,320,18){\sfrac{1}{2}}
\drawbpedge[ELpos=30, ELside=r, ELdist=1, eyo=-1](r2,330,20,q2,320,18){\sfrac{1}{2}}

\drawloop[ELpos=50, ELside=l, ELdist=1, ELside=l,loopCW=y, loopdiam=5, loopangle=90](qX){}



}
\end{picture}
 
\end{center}
\hrule
\caption{{\bf Belief-only is not enough for positive (as well as almost-sure) reachability.}
A one-sided reachability game with reachability objective in which 
player~$1$ is blind and player~$2$ has perfect observation. 
If we consider pure strategies, 
then player~$1$ has a positive (as well as almost-sure) winning strategy, 
but there is no belief-based memoryless positive winning strategy. 
\label{fig:GS09-wrong-positive}}

\end{figure}
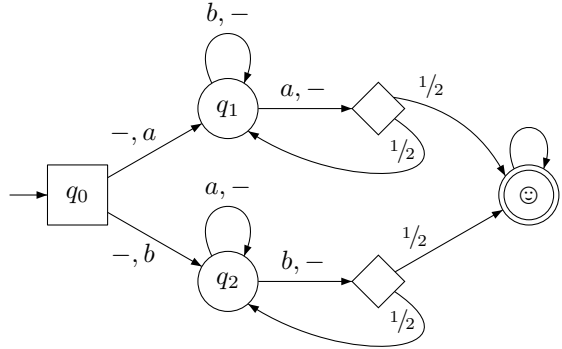

\section{One-sided Games: Player~$1$ Partial and Player~$2$ Perfect}\label{sec:partial-perfect}


In Sections~\ref{sec:partial-perfect} and~\ref{sec:perfect-partial},
we consider one-sided games with partial observation:
one player has perfect observation, and the other player
has partial observation. The player with perfect observation sees the
states visited and the actions played in the game. We present the results
for positive and almost-sure winning for reachability objectives
along with examples that illustrate key elements of the problem such 
as the memory required for winning strategies.

Note that the case of player~$1$ partial and player~$2$ perfect is important
in the context of controller synthesis as it is a conservative approximation 
of two-sided games for player~$1$ (if player~$1$
wins in the one-sided game, then he also wins in the two-sided game).
In the following example we show that for pure strategies \emph{belief-based} 
strategies are not sufficient for positive as well as almost-sure winning.
A strategy is belief-based if its memory relies only on the subset construction,
i.e., the strategy plays only depending on the set of possible current states of 
the game which is called \emph{belief}.


\begin{example}\label{ex:one}
{\bf Belief-only is not enough for positive (as well as almost-sure) reachability.}
Consider the game in \figurename~\ref{fig:GS09-wrong-positive} where player~$1$
is blind (all states have the same observation except the target state, and actions are invisible)
and player~$2$ has perfect observation. Initially, player~$2$ chooses the state
$q_1$ or $q_2$ (which player~$1$ does not see). The belief of player~$1$ is thus 
the set $\{q_1,q_2\}$ (see \figurename~\ref{fig:GS09-wrong-positive-subset-construction}).
We claim that the belief is not a sufficient information to win with a pure strategy for player~$1$
because the belief-based subset construction in \figurename~\ref{fig:GS09-wrong-positive-subset-construction}
suggests that playing always the same action (say $a$) when the belief is $\{q_1,q_2\}$
is an almost-sure winning strategy. However, in the original game this is not even
a positive winning strategy (the counter strategy of player~$2$ is to choose $q_2$ initially).
A winning strategy for player~$1$ is to alternate between $a$ and $b$ when the 
belief is $\{q_1,q_2\}$, which requires to remember more than the belief set.
\ee
\end{example}

We present reductions of the almost-sure and positive winning problem
for reachability objective to the problem of sure-winning in a game of perfect
observation with B\"uchi objective, and reachability objective respectively.
The two reductions are based on the same construction of a game where the state
space $L = \{(s,o) \mid o \subseteq s \subseteq Q \}$ contains the subset 
construction~$s$ enriched with \emph{obligation sets} $o \subseteq s$ which ensure that 
from all states in $s$, the target set $\target$ is reached with positive probability.


\begin{figure}[!tb]
\hrule
\begin{center}
\def\fsize{\normalsize}

\begin{picture}(75,17)(0,0)

{\fsize


\node[Nmarks=i, Nmr=0](q0)(10,7){$q_0$}
\node[Nmarks=n, Nw=10,Nh=10,Nmr=5](q12)(31,7){{\small $q_1,q_2$}}
\rpnode[Nmarks=n](r12)(50,7)(4,4){}
\node[Nmarks=r](qX)(70,7){\smiley}

\drawedge[ELpos=45, ELside=l, ELdist=1](q0,q12){$-,a$}
\drawedge[ELpos=45, ELside=r, ELdist=1](q0,q12){$-,b$}

\drawedge[ELpos=52, ELside=l, ELdist=1, curvedepth=0](q12,r12){$a,-$}
\drawedge[ELpos=52, ELside=r, ELdist=1, curvedepth=0](q12,r12){$b,-$}

\drawedge[ELpos=18, ELside=l, ELdist=.5, curvedepth=2, sxo=2, syo=1.5, exo=-2](r12,qX){\sfrac{1}{2}}
\drawbpedge[ELpos=30, ELside=r, ELdist=1, eyo=-1](r12,330,20,q12,320,18){\sfrac{1}{2}}

\drawloop[ELpos=50, ELside=l, ELdist=1, ELside=l,loopCW=y, loopdiam=5, loopangle=90](qX){}



}
\end{picture}
 
\end{center}
\hrule
\caption{The belief-based subset construction for the reachability game of \figurename~\ref{fig:GS09-wrong-positive}.
Player~$1$ has a pure strategy for positive (as well as almost-sure) winning 
in the subset construction. However, belief-based memoryless pure strategies are
not sufficient in the original game.
\label{fig:GS09-wrong-positive-subset-construction}}

\end{figure}
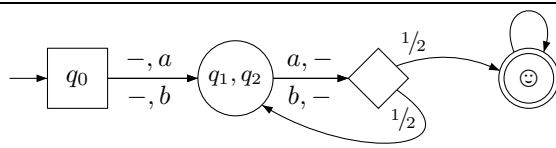


\begin{lemma}\label{lemm:complexity-one-sided-player-one}
Given a one-sided partial-observation stochastic game $G$ with 
player~1 partial and player~2 perfect with a reachability objective for player~1, 
we can construct in time exponential in the size of the game and polynomial in the size of action sets 
a perfect-information deterministic game $H$ with a B\"uchi objective (resp. reachability objective) 
such that player~1 has a pure almost-sure (resp. positive) winning strategy in $G$ iff 
player~1 has a  sure-winning  strategy in $H$.
\end{lemma}


\begin{proof}
We present the construction and the proof in details for almost-sure reachability. 
The construction is the same for positive reachability, and the argument is described 
succinctly afterwards.

\paragraph{\bf Construction.} 
Given $G=\tuple{Q, q_0,\trans}$ over alphabets $A_1,A_2$ and
observation set $\Obs_1$ for player~$1$, with reachability objective $\Reach(\target)$,
we construct the following (deterministic)
game of perfect observation $H =\tuple{L, \l_0, \trans_H}$ over alphabets
$A'_1,A'_2$ with B\"uchi objective $\Buchi(\alpha)$ defined by $\alpha \subseteq L$ where:

\begin{itemize}
\item $L = \{(s,o) \mid o \subseteq s \subseteq Q \}$. Intuitively, 
$s$~is the belief of player~$1$ and~$o$~is a set of obligation states
that ``owe" a visit to $\target$ with positive probability;

\item $\l_0 = (\{q_0\},\{q_0\})$ if $q_0 \not\in \target$, and 
$\l_0 = (\emptyset,\emptyset)$ if $q_0 \in \target$;

\item $A'_1 = A_1 \times 2^Q$. In a pair $(a,u) \in A'_1$, we call $a$ the action,
and $u$ the witness set;

\item $A'_2 = \Obs_1$. In the game $H$, player~$2$ simulate player~$2$'s choice
in game $G$, as well as resolves the probabilistic choices. This amounts to choosing
a possible successor state, and revealing its observation;

\item $\alpha = \{(s,\emptyset) \in L\}$;

\item $\trans_H$ is defined as follows. First, the state $(\emptyset,\emptyset)$ 
is absorbing. Second, in every other state $(s,o) \in L$
the function $\trans_H$ ensures that $(i)$ player~$1$ chooses a pair $(a,u)$ such 
that $\Supp(\trans(q, a, b)) \cap u \neq \emptyset$ for all $q \in o$ and $b \in A_2$,
and $(ii)$ player~$2$ chooses an observation $\gamma \in \Obs_1$ such that 
$\Post_{a,-}(s) \cap \gamma \neq \emptyset$. If a player violates this, 
then a losing absorbing state is reached with probability~$1$. 
Assuming the above condition on $(a,u)$ and $\gamma$ is satisfied, 
define $\trans_H((s,o),(a,u),\gamma)$ as the Dirac distribution on the state 
$(s',o')$ such that:
\begin{itemize}
\item $s' = (\Post_{a,-}(s) \cap \gamma) \setminus \target$; 
\item $o' = s'$ if $o = \emptyset$; and $o' = (\Post_{a,-}(o) \cap \gamma \cap u) \setminus \target$ if $o \neq \emptyset$.
\end{itemize}
\end{itemize}

Note that for every reachable state $(s,o)$ in $H$, there exists a unique 
observation $\gamma \in \Obs_1$ such that $s \subseteq \gamma$ (which we denote
by $\obs_1(s)$). 

We show the following property of this construction. 
Player~$1$ has a pure observation-based
almost-sure winning strategy in~$G$ for the objective $\Reach(\target)$ 
if and only if player~$1$ has a sure winning strategy in~$H$ for 
the objective $\Buchi(\alpha)$.

\paragraph{Mapping of plays.} 
Given a play prefix $\rho_H = (s_0,o_0)(s_1,o_1) \dots (s_k,o_k)$ in $H$ with 
associated actions for player~$1$ of the form $(a_i, \cdot)$ in step $i$ ($0 \leq i < k$),
and a play prefix $\rho_G = q_0 q_1 \dots q_k$ in $G$ with associated 
actions $a'_i$ ($0 \leq i < k$) for player~$1$,  we say that $\rho_G$ is 
\emph{matching} $\rho_H$ if $a_i = a'_i$ for all $0 \leq i < k$, 
and $q_i \in \obs_1(s_i)$ for all $0 \leq i \leq k$.

By induction on the length of $\rho_H$, we show that $(i)$ for each $q_k \in s_k$
there exists a matching play $\rho_G$ (which visits no $\target$-state) 
such that $\Last(\rho_G) = q_k$, and $(ii)$ for all play prefixes $\rho_G$ 
matching $\rho_H$, if $\rho_G$ does not visit any $\target$-state, 
then $\Last(\rho_G) \in s_k$.

For $\abs{\rho_H} = 0$ (i.e., $\rho_H = (s_0,o_0)$ where $(s_0, o_0) = \l_0$) 
it is easy to see that $\rho_G = q_0$
is a matching play with $q_0 \not\in \target$ if and only if $s_0 = o_0 = \{q_0\}$. 
For the induction step, assume that we have constructed matching plays for 
all play prefixes of length~$k-1$, and let $\rho_H = (s_0,o_0)(s_1,o_1) \dots (s_k,o_k)$
be a play prefix of length~$k$ in $H$ with associated actions of the form $(a_i, \cdot)$ in step $i$ ($0 \leq i < k$).
To prove $(i)$, pick $q_k \in s_k$. By definition of $\trans_H$, we have $q_k \in \Post_{a_{k-1},-}(s_{k-1})$,
hence there exists $b \in A_2$ and $q_{k-1} \in s_{k-1}$ such that $q_k \in \Supp(\trans(q_{k-1}, a_{k-1}, b))$.
By induction hypothesis, there exists a play prefix $\rho_G$ in $G$ matching 
$(s_0,o_0) \dots (s_{k-1},o_{k-1})$ and with $\Last(\rho_G) = q_{k-1}$, which we can extend
to $\rho_G.q_{k}$ to obtain a play prefix matching $\rho_H$. To prove $(ii)$, it is easy
to see that every play prefix matching $\rho_H$ is an extension of play prefix 
matching $(s_0,o_0) \dots (s_{k-1},o_{k-1})$ with a non $\target$-state $q_k$ 
in $\gamma_k = \obs_1(s_k)$ and in $\Post_{a_{k-1},-}(s_{k-1})$, therefore 
$q_k \in (\Post_{a_{k-1},-}(s_{k-1}) \cap \gamma_k) \setminus \target = s_k$.

\paragraph{Mapping of strategies, from $G$ to $H$ (ranking argument).}
First, assume that player~$1$ has a pure observation-based almost-sure winning
strategy $\straa$ in~$G$ for the objective $\Reach(\target)$.
We construct an infinite-state MDP $G_{\straa} = \tuple{Q^+, \rho_0, \trans_{\straa}}$ where:

\begin{itemize}
\item $Q^+$ is the set of nonempty finite sequences of states;

\item $\rho_0 = q_0 \in Q$;

\item $\trans_{\straa}: Q^+ \times A_2 \to \dist(Q^+)$ is defined as follows: 
for each $\rho \in Q^+$ and $b \in A_2$, if $\Last(\rho) \not\in \target$ then
$\trans_{\straa}(\rho, b)$ assigns probability
$\trans(\Last(\rho), \straa(\rho), b)(q')$ to each $\rho' = \rho q' \in Q^+$,
and probability $0$ to all other $\rho' \in Q^+$; if $\Last(\rho) \in \target$,
then $\rho$ is an absorbing state;
\end{itemize}

We define a ranking of the reachable states of $G_{\straa}$. Assign rank $0$
to all $\rho \in Q^+$ such that $\Last(\rho) \in \target$. For $i=1,2,\dots$
assign rank $i$ to all non-ranked $\rho$ such that for all player~$2$ actions 
$b \in A_2$, there exists $\rho' \in \Supp(\trans_{\straa}(\rho,b))$ with a rank (and 
thus with a rank smaller than~$i$). We claim that all reachable states of $G_{\straa}$
get a rank. By contradiction, assume that a reachable state $\hat{\rho} = q_0 q_1 \dots q_k$ 
is not ranked (note that $q_i \not\in \target$ for each $0 \leq i \leq k$). 
Fix a strategy $\strab$ for player~$2$ as follows. Since $\hat{\rho}$ is
reachable in $G_{\straa}$, there exist actions $b_0, \dots, b_{k-1}$ such that $q_{i+1} \in \Supp(\trans_{\straa}(q_0 \dots q_i, b_i))$
for all $0 \leq i < k$. Then, define $\strab(q_0 \dots q_i) = b_i$. This ensures 
that $\Last(\hat{\rho})$ is reached with positive probability in $G$ under strategies $\straa$ and $\strab$.
From $\hat{\rho}$, the strategy $\strab$ continues playing as follows. If the current
state $\rho$ is not ranked (which is the case of $\hat{\rho}$), then choose an action~$b$ 
such that all states in $\Supp(\trans_{\straa}(\rho,b))$ are not ranked. The fact that $\rho$ is not ranked
ensures that such an action~$b$ exists. Now, under $\straa$ and $\strab$ all paths
from $\Last(\hat{\rho})$ in $G$ avoid $\target$-sates. Hence the set $\target$
is not reached almost-surely, in contradiction with the fact that $\straa$
is almost-sure winning. Hence all states in $G_{\straa}$ get a rank. We denote
by $\rank(\rho)$ the rank of a reachable state $\rho$ in $G_{\straa}$.

From the strategy $\straa$ and the ranking in $G_{\straa}$, we construct a strategy $\straa'$ in the game $H$
as follows. Given a play $\rho_H = (s_0,o_0)(s_1,o_1) \dots (s_k,o_k)$ in $H$ (with $s_k \neq \emptyset$),
define $\straa'(\rho_H) = (a, u)$ where $a = \straa(\rho_G)$ for a play prefix $\rho_G$ 
matching $\rho_H$ and $u = \{q \in \Supp(\trans(\Last(\rho_G), a, b)) \mid
b \in A_2, \rho_G \text{ is matching } \rho_H \text{ with } \Last(\rho_G) \in o_k
\text{ and } \rank(\rho_G.q) < \rank(\rho_G)\}$ is a witness set which selects
successor states of $o_k$ with decreased rank along each branch of the MDP $G_{\straa}$.

Note that all matching play prefixes $\rho_G$ have the same observation sequence. 
Therefore, the action $a = \straa(\rho_G)$ is unique and well-defined since
$\straa$ is an observation-based strategy. Note also that the pair $(a,u)$
is an allowed choice for player~$1$ by definition of the ranking, and that for each $q\in o_k$, all matching 
play prefixes $\rho_G$ with $\Last(\rho_G) = q$ have the same rank in $G_{\straa}$.
Therefore we abuse notation and write $\rank(q)$ for $\rank(\rho_G)$,
assuming that the set $o_k$ to which $q$ belongs is clear from the context.
Let $\maxrank(o_k) = \max_{q \in o_{k}} \rank(q)$. If $o_k \neq \emptyset$, then
$\maxrank(o_{k+1}) < \maxrank(o_k)$ since $o_{k+1} \subseteq u$ (by definition of $\trans_H$).

\paragraph{Correctness of the mapping.}
We show that $\straa'$ is sure winning for $\Buchi(\alpha)$ in~$H$.
Fix an arbitrary strategy $\strab'$ for player~$2$ in~$H$ and consider
an arbitrary play $\rho_H = (s_0,o_0)(s_1,o_1) \dots$ compatible 
with $\straa'$ and $\strab'$. By the properties of the witness set played
by $\straa'$, for each pair $(s_i,o_i)$ with $o_i \neq \emptyset$, an $\alpha$-pair
$(\cdot, \emptyset)$ is reached within at most $\maxrank(o_i)$ steps. 
And by the properties of the mapping of plays and strategies, if $o_i = \emptyset$
then $o_{i+1} = s_{i+1}$ contains only states from which $\straa$ is almost-sure
winning for $\Reach(\target)$ in~$G$ and therefore have a finite rank, showing
that $\maxrank(o_{i+1})$ is defined and finite. This shows that an $\alpha$-pair
is visited infinitely often in $\rho_H$ and $\straa'$ is sure winning for $\Buchi(\alpha)$.

\paragraph{Mapping of strategies, from $H$ to $G$.}
Given a strategy $\straa'$ in~$H$, we construct a pure observation-based 
strategy $\straa$ in~$G$.

We define $\straa(\rho_G)$ by induction on the length of $\rho_G$.
In fact, we need to define $\straa(\rho_G)$ only for play prefixes $\rho_G$
which are compatible with the choices of $\straa$ for play prefixes
of length smaller than $\abs{\rho_G}$ (the choice of $\straa$ for other 
play prefixes can be fixed arbitrarily). For all such $\rho_G$, our construction
is such that there exists a play prefix $\rho_H = \theta(\rho_G)$ compatible with 
$\straa'$ such that $\rho_G$ is matching $\rho_H$, and if $\straa(\rho_G) = a$
and $\straa'(\rho_H) = (a', \cdot)$, then $a=a'$ $(\star)$. 

We define $\straa$ and $\theta(\cdot)$ as follows.
For $\abs{\rho_G} = 0$ (i.e., $\rho_G = q_0$),    
let $\rho_H = \theta(\rho_G) = (s_0,o_0)$ where $s_0 = o_0 = \{q_0\}$ if $q_0 \not\in \target$, and 
$s_0 = o_0 = \emptyset$ if $q_0 \in \target$,
and let $\straa(\rho_G) = a$ if $\straa'(\rho_H) = (a, \cdot)$. Note that 
property $(\star)$ holds.
For the induction step, let $k\geq 1$ and assume that from every play prefix $\rho_G$
of length smaller than $k$, we have defined $\straa(\rho_G)$ and $\theta(\rho_G)$ satisfying $(\star)$.
Let $\rho_G = q_0 q_1 \dots q_k$ be a play prefix in~$G$ of length $k$. 
Let $\rho_H = \theta(q_0 q_1 \dots q_{k-1})$ and $\gamma_k = \obs_1(q_k)$,
and let $(s_k,o_k)$ be the (unique) successor state in the Dirac distribution
$\trans_H(\Last(\rho_H), \straa'(\rho_H), \gamma_k)$. Note that $q_k \in s_k$.
Define $\theta(\rho_G) = \rho_H.(s_k,o_k)$
and $\straa(\rho_G) = a$ if $\straa'(\rho_H.(s_k,o_k)) = (a, \cdot)$.
Therefore, the property $(\star)$ holds.

Note that the strategy $\straa$ is observation-based because 
if $\obs_1(\rho_G) = \obs_1(\rho'_G)$, then $\theta(\rho_G) = \theta(\rho'_G)$.

\paragraph{Correctness of the mapping.}
If player~$1$ has a sure winning strategy $\straa'$ in~$H$ for the objective $\Buchi(\alpha)$,
then we can assume that $\straa'$ is memoryless (since in perfect-observation deterministic games
with B\"uchi objectives memoryless strategies are sufficient for sure winning~\cite{EmersonJutla91,Thomas97}), 
and we show that the strategy $\straa$ 
defined above is almost-sure winning in~$G$ for the objective $\Reach(\target)$.

Since $\straa'$ is memoryless and sure winning for $\Buchi(\alpha)$, 
in every play compatible with $\straa'$ there are at most $n = \abs{L} \leq 3^{\abs{Q}}$ 
steps between two consecutive visits to an $\alpha$-state.

The properties of matching plays entail that if a play prefix $\rho_G$ compatible
with $\straa$ has no visit to $\target$-states, and $(s,o) = \Last(\theta(\rho_G))$, 
then $\Last(\rho_G) \in s$. Moreover if $s = o$, then under strategy $\straa$ 
for player~$1$ and arbitrary strategy $\strab$ for player~$2$, there is a way to 
fix the probabilistic choices such that all plays extension of $\rho_G$ visit
a $\target$-state. To see this, consider the probabilistic choices given at each 
step by the witness component $u$ of the action $(\cdot, u)$ played by $\straa'$. 
By the definition of the mapping of plays and of the transition function in~$H$, 
it can be shown that if $(s_i,o_i)(s_{i+1},o_{i+1}) \dots (s_k,o_k)$ is a play fragment
of $\theta(\rho_G)$ (hence compatible with $\straa'$) where $s_i = o_i$ and $o_j \neq \emptyset$ for all $i \leq j < k$, 
then the ``owe" set $o_k$ is the set of all states that can be reached in~$G$ from states $s_i$ 
along a path which is compatible with both the
action played by the strategy $\straa'$ (and $\straa$) and the 
probabilistic choices fixed by $\straa'$, and visits no $\target$-states.
Since the ``owe" set gets empty within at most $n$ steps regardless of the strategy 
of player~$2$, all paths compatible with the probabilistic choices must visit
an $\target$-state. This shows that under any player~$2$ strategy, within $n$
steps, a $\target$-state is visited with probability at least $r^n$ where $r>0$ 
is the smallest non-zero probability occurring in $G$. Therefore, the probability of not
having visited a $\target$-state after $z\cdot n$ steps is at most $(1-r^n)^z$
which vanishes for $z \to \infty$ since $r^n > 0$. Hence, against arbitrary strategy of player~$2$,
the strategy $\straa$ ensures the objective $\Reach(\target)$ with probability~$1$.

\paragraph{Argument for positive reachability.}
The proof for positive reachability follows the same line as for almost-sure
reachability, with the following differences. The construction of the game
of perfect information~$H$ is now interpreted as a reachability game with 
objective $\Reach(\alpha)$. The mapping of plays is the same as above.
In the mapping of strategies from $G$ to $H$, we use the same ranking 
construction, but we only claim that the initial state gets a rank. 
The argument is that if the initial state would get no rank, then player~$2$
would have a strategy to ensure that all paths avoid the target states, in contradiction
with the fact that player~$1$ has fixed a positive winning strategy.
The rest of the proof is analogous to the case of almost-sure reachability.
\qed
\end{proof}

It follows from the construction in the proof of Lemma~\ref{lemm:complexity-one-sided-player-one}
that pure strategies with exponential memory are sufficient for 
positive (as well as almost-sure) winning, and 
the exponential lower bound follows from the special case of non-stochastic games~\cite{BD08}.
Lemma~\ref{lemm:complexity-one-sided-player-one} also gives EXPTIME upper bound
for the problem since perfect-observation B\"uchi games can be solved in polynomial time~\cite{Thomas97}. 
The EXPTIME-hardness follows from the sure winning problem for non-stochastic games~\cite{Reif84}, 
where pure almost-sure (positive) winning strategies coincide with sure winning strategies.
We have the following theorem summarizing the results.

\begin{theorem}\label{theo:complexity-one-sided-player-one}
Given one-sided partial-observation stochastic games with player~1 partial
and player~2 perfect, the following assertions hold for
reachability objectives for player~1:
\begin{enumerate}
\item \emph{(Memory complexity).} Belief-based pure strategies are not sufficient
both for positive and almost-sure winning; exponential memory is necessary and
sufficient both for positive and almost-sure winning for pure strategies.

\item \emph{(Algorithm).} The problems of deciding the existence of a 
pure almost-sure and a pure positive winning strategy can be solved 
in time exponential in the state space of the game and polynomial in 
the size of the action sets.

\item \emph{(Complexity).} The problems of deciding the existence of a pure 
almost-sure and a pure positive winning strategy are EXPTIME-complete.
 
\end{enumerate}
\end{theorem}



\paragraph{Symbolic algorithms.}
The exponential B\"uchi (or reachability) game constructed in the proof of 
Theorem~\ref{theo:complexity-one-sided-player-one} can be solved by computing
classical fixpoint formulas~\cite{EmersonJutla91}. However, it is not necessary to 
construct the exponential game structure explicitly. Instead, we can exploit
the structure induced by the pre-order $\preceq$ defined by 
$(s,o) \preceq (s',o')$ if $(i)$ $s\subseteq s'$, $(ii)$ $o\subseteq o'$,
and $(iii)$ $o=\emptyset$ iff $o' = \emptyset$. Intuitively, if a state $(s',o')$
is winning for player~$1$, then all states $(s,o) \preceq (s',o')$ are also winning
because they correspond to a better belief and a looser obligation. 
Hence all sets computed by the fixpoint algorithm are downward-closed
and thus they can be represented symbolically by the antichain of their maximal 
elements (see~\cite{CDHR07} for details related to antichain algorithms). 
This technique provides a symbolic algorithm without explicitly constructing the exponential game.

\section{One-sided Games: Player~$1$ Perfect and Player~$2$ Partial}\label{sec:perfect-partial}

Recall that we are interested in finding a pure winning strategy for
player~$1$. Therefore, when we construct counter-strategies for player~$2$, we always
assume that player~$1$ has already fixed a pure strategy. This is important
for the way the belief of player~$2$ is updated. Although player~$2$ does not
have perfect information about the actions played by player~$1$, the belief
of player~$2$ can be updated according to the precise actions of player~$1$ because 
the response and the counter-strategy of player~$2$ is designed after
player~$1$ has fixed a strategy.

\subsection{Lower bound on memory}

We present the following examples to illustrate two properties of the problem.

\begin{figure}[!tb]
\hrule
\begin{center}
\def\fsize{\normalsize}

\begin{picture}(90,60)(0,-1)

{\fsize


\node[Nmarks=i, Nmr=0](q0)(5,27){$q_0$}

\rpnode[Nmarks=n](r1)(25,37)(4,3.5){}
\rpnode[Nmarks=n](r2)(25,17)(4,3.5){}

\node[Nmarks=n](q1)(45,47){$q_1$}
\node[Nmarks=n](q2)(45,27){$q_2$}
\node[Nmarks=n](q3)(45,7){$q_3$}

\node[Nmarks=n, Nmr=0](q4)(65,37){$q_4$}
\node[Nmarks=n, Nmr=0](q5)(65,17){$q_5$}

\node[Nmarks=r, rdist=0.8](qX)(85,27){\smiley}

\drawedge[ELpos=49, ELside=l, ELdist=.5, curvedepth=0](q0,r1){$-,a$}
\drawedge[ELpos=50, ELside=r, ELdist=1, curvedepth=0](q0,r2){$-,b$}

\drawedge[ELpos=36, ELside=l, ELdist=1, curvedepth=0](r1,q1){\sfrac{1}{2}}
\drawedge[ELpos=36, ELside=r, ELdist=1, curvedepth=0](r1,q2){\sfrac{1}{2}}
\drawedge[ELpos=36, ELside=l, ELdist=1, curvedepth=0](r2,q2){\sfrac{1}{2}}
\drawedge[ELpos=36, ELside=r, ELdist=1, curvedepth=0](r2,q3){\sfrac{1}{2}}

\drawedge[ELpos=40, ELside=l, ELdist=-.5, curvedepth=0](q1,q4){$\begin{array}{cl}a\!&,- \\[-2pt] b\!&,-\end{array}$} 
\drawedge[ELpos=47, ELside=l, ELdist=0, curvedepth=0](q2,q4){$a,-$}
\drawedge[ELpos=48, ELside=r, ELdist=0.5, curvedepth=0](q2,q5){$b,-$}
\drawedge[ELpos=40, ELside=r, ELdist=-.5, curvedepth=0](q3,q5){$\begin{array}{cl}a\!&,- \\[-2pt] b\!&,-\end{array}$}

\drawedge[ELpos=45, ELside=l, ELdist=.5, curvedepth=0](q4,qX){$-,a$}
\drawedge[ELpos=45, ELside=r, ELdist=.5, curvedepth=0](q5,qX){$-,b$}


\node[Nframe=n, Nw=0,Nh=0](tmp1)(63,53){}
\node[Nframe=n, Nw=0,Nh=0](tmp2)(7,53){}
\drawline[AHnb=0, arcradius=2](65,41)(65,53)(60,53)
\drawedge[AHnb=0, ELpos=55, ELside=r, ELdist=1, curvedepth=0](tmp1,tmp2){$-,b$}
\drawline[AHnb=1, arcradius=2](10,53)(5,53)(5,31)

\node[Nframe=n, Nw=0,Nh=0](tmp1)(63,1){}
\node[Nframe=n, Nw=0,Nh=0](tmp2)(7,1){}
\drawline[AHnb=0, arcradius=2](65,13)(65,1)(60,1)
\drawedge[AHnb=0, ELpos=55, ELside=l, ELdist=1, curvedepth=0](tmp1,tmp2){$-,a$}
\drawline[AHnb=1, arcradius=2](10,1)(5,1)(5,23)


\drawloop[ELpos=50, ELside=l, ELdist=1, ELside=l,loopCW=y, loopdiam=7, loopangle=90](qX){}



}
\end{picture}
\end{center}
\hrule
\caption{{\bf Remembering the belief of player~$2$ is necessary.}
A one-sided reachability game where player~$1$ (round states) has perfect observation, 
player~$2$ (square states) is blind. Player~$1$ has a pure almost-sure winning strategy 
that depends on the belief of player~$2$ (in $q_2$), but no pure memoryless strategy is almost-sure winning.
\label{fig:belief-necessary}
}

\end{figure}
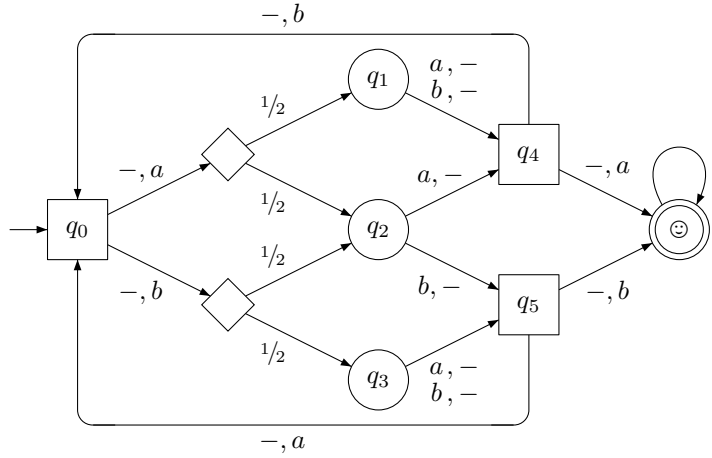

\begin{example}
{\bf Remembering the belief of player~$2$ is necessary.}
We present an example of a game where player~$1$ has perfect observation
but needs to remember the belief of player~$2$ to ensure positive or almost-sure reachability.
The game is shown in \figurename~\ref{fig:belief-necessary}. The target is $\target = \{q_{\smiley}\}$. Player~$2$ is blind.
If player~$2$ chooses $a$ in the initial state $q_0$, then his belief will be $\{q_1,q_2\}$,
and if he plays $b$, then his belief will be $\{q_2,q_3\}$. In $q_2$, the choice 
of player~$1$ depends on the belief of player~$2$. If the belief is $\{q_1,q_2\}$,
then playing $a$ in $q_2$ is not a good choice because the belief of player~$2$ would be
$\{q_4\}$ and player~$2$ could surely avoid $q_{\smiley}$ by further playing $b$. 
For symmetrical reasons, if the belief of player~$2$ is $\{q_2,q_3\}$ in $q_2$, then
playing $b$ is not a good choice for player~$1$. Therefore, there is no positively winning
\emph{memoryless} strategy for player~$1$. However, we show that there exists an almost-sure winning
\emph{belief-based} strategy for player~$1$ as follows: in $q_2$, play $b$ if the belief of player~$2$ is $\{q_1,q_2\}$,
and play $a$ if the belief of player~$2$ is $\{q_2,q_3\}$. Note that player~$1$
has perfect observation and thus can observe the actions of player~$2$.
This ensures the next belief of player~$2$ to be $\{q_3,q_4\}$ and therefore no matter the next action
of player~$2$, the state $q_{\smiley}$ is reached with probability $\frac{1}{2}$.
Repeating this strategy ensures to reach $q_{\smiley}$ with probability $1$.
\ee
\end{example}

\begin{figure}[!tb]
\hrule
\begin{center}
\def\fsize{\normalsize}

\begin{picture}(130,40)(0,0)

{\fsize


\rpnode[Nmarks=i, Nmr=0](qI)(10,20)(4,4){$q_I$}

\node[Nmarks=n](qup)(30,30){$L$}     
\node[Nmarks=n](qdown)(30,10){$R$}   

\node[Nmarks=n, dash={1.5}0, Nw=12, Nh=32, Nmr=3](box)(30,20){}

\node[Nmarks=n](qn)(50,20){$q_{n}$}
\node[Nmarks=n](qn-1)(70,20){$q_{n-1}$}
\node[Nframe=n](dummy)(90,20){$\ \ \dots$}
\node[Nmarks=n](q1)(100,20){$q_{1}$}
\node[Nmarks=r](q0)(120,20){$q_{0}$}  

\drawedge[ELpos=45, ELside=l, ELdist=.5](qI,qup){$\sfrac{1}{2}$}
\drawedge[ELpos=45, ELside=r, ELdist=.5](qI,qdown){$\sfrac{1}{2}$}

\drawedge[ELpos=45, ELside=l, ELdist=.5](qup,qn){$b,-$}
\drawedge[ELpos=45, ELside=r, ELdist=.5](qdown,qn){$b,-$}

\drawedge[ELpos=45, ELside=r, ELdist=.5, curvedepth=-10](qup,qI){$a,-$}
\drawedge[ELpos=45, ELside=l, ELdist=.5, curvedepth=10](qdown,qI){$a,-$}

\drawloop[ELpos=51, ELside=l, ELdist=0, ELside=l,loopCW=y, loopdiam=6, loopangle=90](qn){$\begin{array}{l}a,a\\b,b\end{array}$}
\drawedge[ELpos=50, ELside=l, ELdist=.5, curvedepth=0](qn,qn-1){$a,b$}
\drawedge[ELpos=50, ELside=r, ELdist=.5, curvedepth=0](qn,qn-1){$b,a$}
\drawbpedge[ELpos=36, ELside=l, ELdist=-.5, curvedepth=8](qn-1,240,10,qn,295,15){$\begin{array}{l}a,a\\b,b\end{array}$}

\drawedge[ELpos=50, ELside=l, ELdist=.5, curvedepth=0](qn-1,dummy){$a,b$}
\drawedge[ELpos=50, ELside=r, ELdist=.5, curvedepth=0](qn-1,dummy){$b,a$}

\drawedge[ELpos=50, ELside=l, ELdist=.5, curvedepth=0](q1,q0){$a,b$}
\drawedge[ELpos=50, ELside=r, ELdist=.5, curvedepth=0](q1,q0){$b,a$}
\drawbpedge[ELpos=30, ELside=l, ELdist=-.5, curvedepth=18](q1,225,20,qn,280,30){$\begin{array}{l}a,a\\b,b\end{array}$}

\node[Nframe=n](dummy)(80,11){$\dots$}

\drawloop[ELpos=51, ELside=l, ELdist=0, ELside=l,loopCW=y, loopdiam=6, loopangle=90](q0){}



}
\end{picture}
\end{center}
\hrule
\caption{A one-sided reachability game $L_n$ with reachability objective in which 
player~$1$ is has perfect observation and player~$2$ is blind. 
Player~$1$ needs exponential memory to win positive reachability.
\label{fig:gadget-mult-div}}

\end{figure}
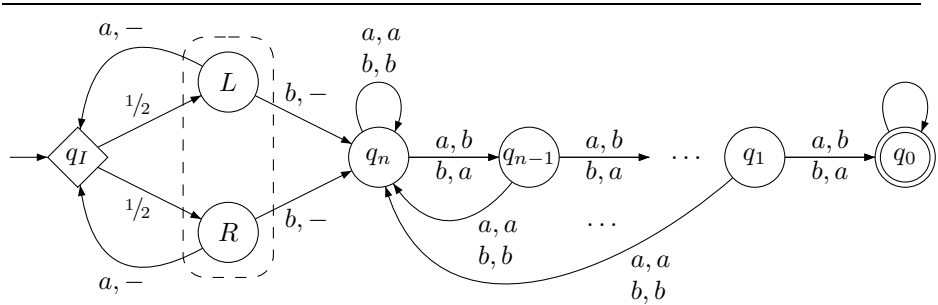

\begin{example}
{\bf Memory of non-elementary size may be necessary for positive and almost-sure reachability.}
We show that player~$1$ may need memory of non-elementary size to 
win positively (as well as almost-surely) 
in a reachability game. 
We present a family of one-sided games $G_n$ where player~$1$ has perfect observation,
and player~$2$ has partial observation both about the state of the game, and the 
actions played by player~$1$. We explain the example step by step. The key idea
of the example is that the winning strategy of player~$1$ in game $G_n$ 
will need to simulate a counter systems (with $n$ integer-valued counters) where the operations
on counters are \emph{increment} and \emph{division by $2$} (with round down),
and to reach strictly positive counter values.

\smallskip\noindent{\em Counters.}
First, we use a simple example to show that counters appear naturally in the
analysis of the game under pure strategies.

Consider the family of games $(L_n)_{n \in \nat}$ shown in \figurename~\ref{fig:gadget-mult-div},
where the reachability objective is $\Reach(\{q_0\})$.
In the first part, the states $L$ and $R$ are indistinguishable for player~$2$.
Consider the strategy of player~$1$ that plays $b$ in $L$ and $R$. 
Then, the state $q_n$ is reached by two play prefixes $\rho_{up} = q_I L q_n$ 
and $\rho_{dw} = q_I R q_n$ that player~$2$ cannot distinguish. 
Therefore, player $2$ has to play the same
action in both play prefixes, while perfectly-informed player~$1$ can play different actions.
In particular, if player~$1$ plays $a$ in $\rho_{up}$ and $b$ in $\rho_{dw}$, 
then no matter the action chosen by player~$2$ the state $q_{n-1}$ is reached with
positive probability. However, because only one play prefix reaches $q_{n-1}$, 
this strategy of player~$1$ cannot ensure to reach $q_{n-2}$ with positive probability.

Player~$1$ can ensure to reach $q_{n-2}$ (and $q_0$) with positive probability with the following
exponential-memory strategy. For the first $n-1$ visits to either $L$ or $R$,
play $b$, and on the $n$th visit, play $a$. This strategy produces $2^n$ different play prefixes from $q_I$
to $q_n$, each with probability $\frac{1}{2^n}$. Considering the mapping $L \mapsto a$,
$R \mapsto b$, each such play prefix $\rho$ is mapped to a sequence $w_\rho$ of length $n$ over $\{a,b\}$
(for example, the play prefix $q_I L q_I R q_I L q_n$ is mapped to $aba$).
The strategy of player~$1$ is to play the sequence $w_\rho$ in the next $n$ steps after $\rho$.
This strategy ensures that for all $0 \leq i \leq n$, there are $2^i$ play prefixes 
which reach $q_i$ with positive probability, all being indistinguishable for player~$2$.
The argument is an induction on $i$. The claim is true for $i=n$, and if it holds for $i=k$,
then no matter the action chosen by player~$2$ in $q_{k}$, the state $q_{k-1}$ is reached
with positive probability by half of the $2^k$ play prefixes, i.e. $2^{k-1}$ play prefixes. 
This establishes the claim. As a consequence,
one play prefix reaches $q_0$ with positive probability. This strategy requires exponential
memory, and an inductive argument shows that this memory is necessary because 
player~$1$ needs to have at least $2$ play prefixes that are indistinguishable 
for player~$2$ in state $q_1$, and at least $2^i$ play prefixes in $q_i$ for all $0 \leq i \leq n$.

\begin{figure}[!tb]
\hrule
\begin{center}
\def\fsize{\normalsize}

\begin{picture}(130,32)(0,0)

{\fsize


\node[Nmarks=i](q4)(10,10){$q_{4}$}
\node[Nmarks=n](q3)(31,10){$q_{3}$}
\node[Nmarks=n](q2)(55,10){$q_{2}$}
\node[Nmarks=n](q1)(82,10){$q_{1}$}
\node[Nmarks=r](q0)(118,10){$q_{0}$}

\node[Nmarks=i, ExtNL=y, NLangle=270, NLdist=1.5](l4)(10,10){$[0,0,0,2^{2059}]$}
\node[Nmarks=n, ExtNL=y, NLangle=270, NLdist=3](l3)(31,10){$[0,0,2^{11},2^{11}]$}
\node[Nmarks=n, ExtNL=y, NLangle=270, NLdist=1.5](l2)(55,10){$[0,2^3,2^3,2^3]$}
\node[Nmarks=n, ExtNL=y, NLangle=270, NLdist=3](l1)(82,10){$[2,2,2,2]$}
\node[Nmarks=n, ExtNL=y, NLangle=270, NLdist=3](l0)(118,10){$[1,1,1,1]$} 

\drawedge[ELpos=50, ELside=l, ELdist=.5](q4,q3){}
\drawedge[ELpos=50, ELside=l, ELdist=.5](q3,q2){}
\drawedge[ELpos=50, ELside=l, ELdist=.5](q2,q1){}
\drawedge[ELpos=50, ELside=r, ELdist=1](q1,q0){$(\div 2,\div 2,\div 2,\div 2)$}

\drawloop[ELpos=50, ELside=l, ELdist=1, ELside=l,loopCW=y, loopdiam=6, loopangle=90](q4)
{$\begin{array}{c}2^{11} \cdot 2^{2^{11}} \\(\cdot,\cdot,\cdot,+1)\end{array}$}
\drawloop[ELpos=50, ELside=l, ELdist=1, ELside=l,loopCW=y, loopdiam=6, loopangle=90](q3)
{$\begin{array}{c}2^{3} \cdot 2^{2^3} \\(\cdot,\cdot,+1,\div 2)\end{array}$}
\drawloop[ELpos=50, ELside=l, ELdist=1, ELside=l,loopCW=y, loopdiam=6, loopangle=90](q2)
{$\begin{array}{c}2 \cdot 2^{2} \\(\cdot,+1,\div 2,\div 2)\end{array}$}
\drawloop[ELpos=50, ELside=l, ELdist=1, ELside=l,loopCW=y, loopdiam=6, loopangle=90](q1)
{$\begin{array}{c}2 \\(+1,\div 2,\div 2,\div 2)\end{array}$}



}
\end{picture}
\end{center}
\hrule
\caption{A family $(C_n)_{n \in \nat}$ of counter systems with $n$ counters and
$n+1$ states where the shortest execution to
reach $(q_0, k_1, \dots, k_n)$ with positive counters (i.e., $k_i > 0$ for all $1\leq i \leq n$)
from $(q_n, 0, \dots,0)$ is of non-elementary length. The numbers above the self-loops
show the number of times each self-loop is taken along the shortest execution.
\label{fig:non-elementary-counters}}
\end{figure}
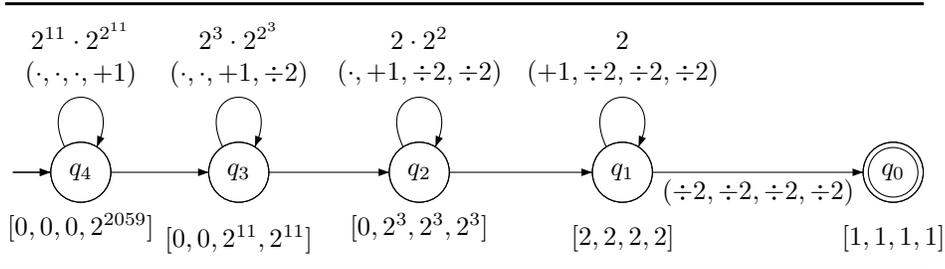

\smallskip\noindent{\em Non-elementary counters.}
Now, we present a family $C_n$ of counter systems where the shortest execution
is of non-elementary length (specifically, the shortest length is greater than 
a tower $2^{2^{\cdot^{\cdot^{2}}}}$ of exponentials of height $n$). 
The counter system $C_4$ (for $n=4$) is shown in \figurename~\ref{fig:non-elementary-counters}.
The operations on counters can be \emph{increment} ($+1$), \emph{division by $2$} ($\div 2$),
and \emph{idle} ($\cdot$).
In general, $C_n$ has $n$ counters $c_1, \dots, c_n$ and $n+1$ states $q_0, \dots, q_n$.
In state $q_i$ of $C_n$ ($0 \leq i \leq n$), the counter $c_i$ can 
be incremented and at the same time all the counters $c_j$ for $j>i$ are divided by $2$.
From $q_n$, to reach $q_0$ with strictly positive counters (i.e., all counters
have value at least $1$), we show that it is necessary to execute the self-loop
on state $q_n$ a non-elementary number of times. In \figurename~\ref{fig:non-elementary-counters},
the numbers above the self-loops show the number of times they need to be executed.
When leaving $q_1$, the counters need to have value at least $2$ in order to survive
the transition to $q_0$ which divides all counters by $2$. 
Since the first counter can be incremented only in state $q_1$, the self-loop
in $q_1$ has to be executed $2$ times. Hence, when leaving $q_2$, the other counters
need to have value at least $2 \cdot 2^2 = 2^3$ in order to survive the self-loops in $q_1$. 
Therefore, the self-loop in $q_2$ is executed $2^3$ times. And so on.
In general, if the self-loop on state $q_i$ is executed $k$ times (in order to get $c_i=k$), 
then the counters $c_{i+1}, \dots, c_n$ need to have value $k \cdot 2^k$ when entering $q_i$
(in order to guarantee a value at least $k$ of these counters). In $q_n$, 
the last counter $c_n$ needs to have value $f^n(1)$ where $f^n$ is the $n$th
iterate of the function $f: \nat \to \nat: x \mapsto x \cdot 2^x$. This value
is greater than a tower of exponentials of height $n$.

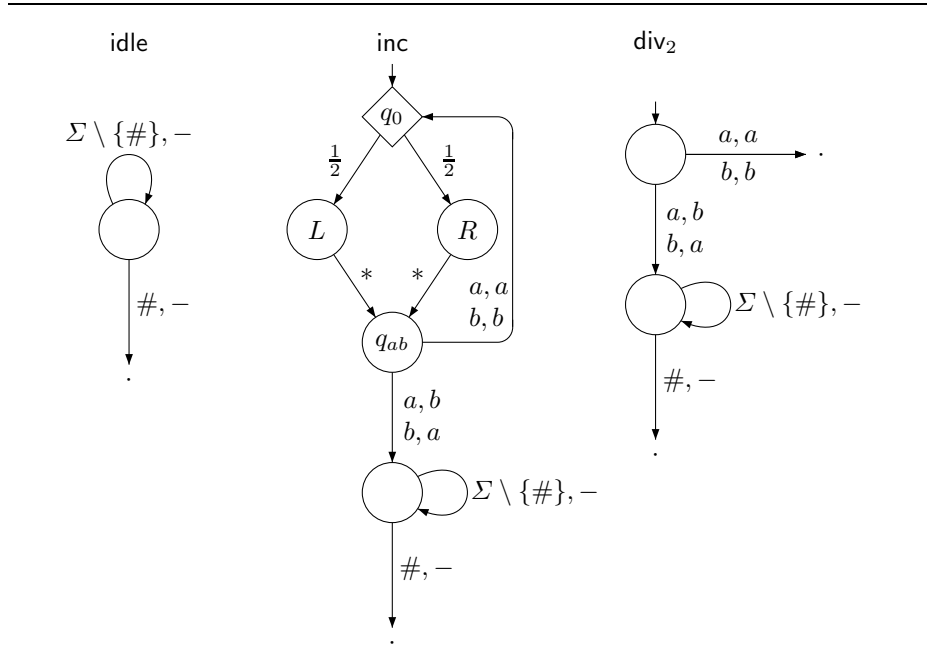
\begin{figure}[!tb]
\hrule
\begin{center}
\def\fsize{\normalsize}

\begin{picture}(120,85)(0,0)

{\fsize


\node[Nframe=n](idle)(15,80){\idle}

\node[Nmarks=n](q)(15,55){}
\node[Nframe=n, Nh=4, Nw=4](qbis)(15,35){$\cdot$}

\drawloop[ELpos=50, ELside=l, ELdist=1, ELside=l,loopCW=y, loopdiam=6, loopangle=90](q)
{$\Sigma\setminus \{\#\}, -$}

\drawedge[ELpos=50, ELside=l, ELdist=1, curvedepth=0](q,qbis){$\#,-$}

\node[Nframe=n](inc)(50,80){\inc}

\rpnode[Nmarks=i, iangle=90, ilength=3](r6)(50,70)(4,4){$q_0$}

\node[Nmarks=n](rl)(40,55){$L$}
\node[Nmarks=n](rr)(60,55){$R$}

\node[Nmarks=n](r5)(50,40){$q_{ab}$}
\node[Nmarks=n](r4)(50,20){}
\node[Nframe=n, Nh=4, Nw=4](r4bis)(50,0){$\cdot$}

\drawedge[ELpos=50, ELside=r, ELdist=1, curvedepth=0](r6,rl){$\frac{1}{2}$}
\drawedge[ELpos=50, ELside=l, ELdist=1, curvedepth=0](r6,rr){$\frac{1}{2}$}

\drawedge[ELpos=50, ELside=l, ELdist=.5, curvedepth=0](rl,r5){*}
\drawedge[ELpos=50, ELside=r, ELdist=.5, curvedepth=0](rr,r5){*}

\node[Nframe=n, Nw=0,Nh=0](tmp1)(66,42){}
\node[Nframe=n, Nw=0,Nh=0](tmp2)(66,68){}
\drawline[AHnb=0, arcradius=2](54,40)(66,40)(66,43)
\drawedge[AHnb=0, ELpos=12, ELside=l, ELdist=0, curvedepth=0](tmp1,tmp2){$\begin{array}{l}a,a\\b,b\end{array}$}
\drawline[AHnb=1, arcradius=2](66,67)(66,70)(54,70)

\drawedge[ELpos=50, ELside=l, ELdist=1, curvedepth=0](r5,r4){$\begin{array}{l}a,b\\b,a\end{array}$}

\drawloop[ELpos=50, ELside=l, ELdist=0.5, ELside=l,loopCW=y, loopdiam=6, loopangle=0](r4)
{$\Sigma\setminus \{\#\}, -$}

\drawedge[ELpos=50, ELside=l, ELdist=1, curvedepth=0](r4,r4bis){$\#,-$}

\node[Nframe=n](div)(85,80){${\sf div}_2$}

\node[Nmarks=i, iangle=90, ilength=3](s2)(85,65){}
\node[Nframe=n, Nh=4, Nw=4](s2bis)(107,65){$\cdot$}
\node[Nmarks=n](s1)(85,45){}
\node[Nframe=n, Nh=4, Nw=4](s1bis)(85,25){$\cdot$}

\drawedge[ELpos=50, ELside=l, ELdist=1, curvedepth=0](s2,s2bis){$a,a$}
\drawedge[ELpos=50, ELside=r, ELdist=1, curvedepth=0](s2,s2bis){$b,b$}

\drawedge[ELpos=50, ELside=l, ELdist=1, curvedepth=0](s2,s1){$\begin{array}{l}a,b\\b,a\end{array}$}

\drawloop[ELpos=50, ELside=l, ELdist=.5, ELside=l,loopCW=y, loopdiam=6, loopangle=0](s1)
{$\Sigma\setminus \{\#\}, -$}

\drawedge[ELpos=50, ELside=l, ELdist=1, curvedepth=0](s1,s1bis){$\#,-$}



}
\end{picture}
\end{center}
\hrule
\caption{Gadgets to simulate {\sf idle}, {\sf increment}, and {\sf division by 2}.
\label{fig:gadget-non-elementary}}
\end{figure}

\smallskip\noindent{\em Gadgets for increment and division.}
In \figurename~\ref{fig:gadget-non-elementary}, we show the gadgets that are used
to simulate operations on counters. The gadgets are game graphs where the 
player-$1$ actions $a,b$ are indistinguishable for player~$2$ (but player~$2$ can
observe and distinguish the action $\#$). The actions $a,b$ are used by player~$1$ to 
simulate the operations on the counters. The $\#$ is used to simulate 
the transitions from state $q_i$ to $q_{i-1}$ of the counter system of~\figurename~\ref{fig:non-elementary-counters}.
All states of the gadgets have the same observation for player~$2$.
Recall that player~$1$ has perfect observation.

The idle gadget is straightforward. The actions $a,b$ have no effect.
In the other gadgets, the value of the counters is represented by the number of
paths that are indistinguishable for player~$2$, and that end up in the entry
state of the gadget (for the value of the counter before the operation)
or in the exit state (for the value of the counter after the operation).

Consider the division gadget ${\sf div}_2$. If player $2$ plays an action that matches
the choice of player~$1$, then the game leaves the gadget and the transition
will go to the initial state of the game we construct (which is shown on \figurename~\ref{fig:non-elementary-game}).
Otherwise, the action of player~$2$ does not match the action of player~$1$
and the play reaches the exit state of the gadget. Let $k$ be the number of
indistinguishable\footnote{In the rest of this section, the word \emph{indistinguishable}
means \emph{indistinguishable for player~$2$}.} paths in the entry state of the gadget. By playing $a$
after $k_1$ such paths and $b$ after $k_2$ paths (where $k_1 + k_2 = k$),
player~$1$ ensures that $\min\{k_1,k_2\}$ indistinguishable paths reach the
exit state of the gadget (because in the worst case, player~$2$ can choose his action
to match the action of player~$1$ over $\max\{k_1,k_2\}$ paths). Hence, player~$1$
can ensure that $\lfloor \frac{k}{2} \rfloor$ indistinguishable paths get to
the exit state. In the game of \figurename~\ref{fig:non-elementary-game},
the entry and exit state of division gadgets are merged. The argument still holds.

Consider the increment gadget $\inc$ on \figurename~\ref{fig:gadget-non-elementary}. 
We use this gadget with the assumption that the entry state is not reached
by more than one indistinguishable path. This will be the case in the 
game of \figurename~\ref{fig:non-elementary-game}. Player~$1$ can achieve
$k$ indistinguishable paths in the exit state as follows. In state $q_{ab}$,
play action $a$ if the last visited state is $q_L$, and play action $b$ if 
the last visited state is $q_R$. No matter the choice of player~$1$, one path
will reach the exit state, and the other path will get to the entry state.
Repeating this scenario $k$ times gives $k$ paths in the exit state. 
We show that there is essentially no faster way to obtain $k$ paths in the exit state.
Indeed, if player~$1$ chooses the same action (say $a$) after the two paths ending up in $q_{ab}$, then 
against the action $b$ from player~$2$, two paths reach the exit state, and
no state get to the entry state. Then, player~$1$ can no longer increment
the  number of paths. Therefore, to get $k$ paths in the exit state, the fastest
way is to increment one by one up to $k-2$, and then get $2$ more paths as a last
step. Note that it is not of the interest of player~$2$ to match the action 
of player~$1$ if player~$1$ plays the same action, because this would double
the number of paths.

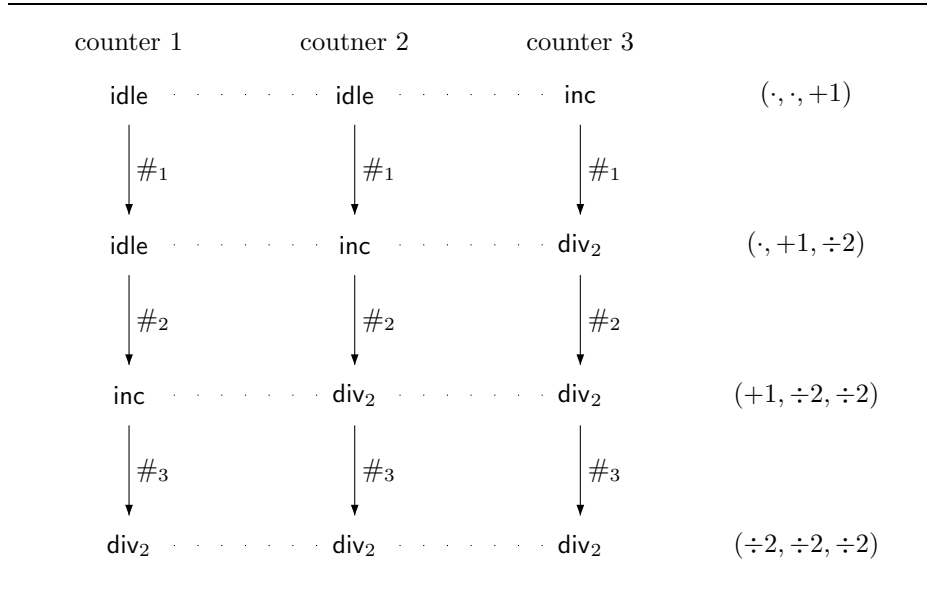
\begin{figure}[!tb]
\hrule
\begin{center}
\def\fsize{\normalsize}

\begin{picture}(120,77)(0,0)

{\fsize


\node[Nframe=n](c1)(15,72){counter 1}
\node[Nframe=n](c2)(45,72){coutner 2}
\node[Nframe=n](c3)(75,72){counter 3}

\node[Nframe=n](q4)(15,65){\idle}
\node[Nframe=n](q3)(15,45){\idle}
\node[Nframe=n](q2)(15,25){\inc}
\node[Nframe=n](q1)(15,5){${\sf div}_2$}

\drawedge[ELpos=50, ELside=l, ELdist=1, curvedepth=0](q4,q3){$\#_1$}
\drawedge[ELpos=50, ELside=l, ELdist=1, curvedepth=0](q3,q2){$\#_2$}
\drawedge[ELpos=50, ELside=l, ELdist=1, curvedepth=0](q2,q1){$\#_3$}

\node[Nframe=n](r4)(45,65){\idle}
\node[Nframe=n](r3)(45,45){\inc}
\node[Nframe=n](r2)(45,25){${\sf div}_2$}
\node[Nframe=n](r1)(45,5){${\sf div}_2$}

\drawedge[ELpos=50, ELside=l, ELdist=1, curvedepth=0](r4,r3){$\#_1$}
\drawedge[ELpos=50, ELside=l, ELdist=1, curvedepth=0](r3,r2){$\#_2$}
\drawedge[ELpos=50, ELside=l, ELdist=1, curvedepth=0](r2,r1){$\#_3$}

\node[Nframe=n](s4)(75,65){\inc}
\node[Nframe=n](s3)(75,45){${\sf div}_2$}
\node[Nframe=n](s2)(75,25){${\sf div}_2$}
\node[Nframe=n](s1)(75,5){${\sf div}_2$}

\drawedge[ELpos=50, ELside=l, ELdist=1, curvedepth=0](s4,s3){$\#_1$}
\drawedge[ELpos=50, ELside=l, ELdist=1, curvedepth=0](s3,s2){$\#_2$}
\drawedge[ELpos=50, ELside=l, ELdist=1, curvedepth=0](s2,s1){$\#_3$}

\node[Nframe=n](r4)(105,65){$(\cdot,\cdot,+1)$}
\node[Nframe=n](r3)(105,45){$(\cdot,+1,\div 2)$}
\node[Nframe=n](r2)(105,25){$(+1,\div 2,\div 2)$}
\node[Nframe=n](r1)(105,5){$(\div 2,\div 2,\div 2)$}

\gasset{dash={.15 3}0}

\drawline[AHnb=0, arcradius=2](20,65)(40,65)
\drawline[AHnb=0, arcradius=2](20,45)(40,45)
\drawline[AHnb=0, arcradius=2](20,25)(40,25)
\drawline[AHnb=0, arcradius=2](20,5)(40,5)

\drawline[AHnb=0, arcradius=2](50,65)(70,65)
\drawline[AHnb=0, arcradius=2](50,45)(70,45)
\drawline[AHnb=0, arcradius=2](50,25)(70,25)
\drawline[AHnb=0, arcradius=2](50,5)(70,5)



}
\end{picture}
\end{center}
\hrule
\caption{Abstract view of the game in \figurename~\ref{fig:non-elementary-game}
as a 3-counter system.
\label{fig:non-elementary-game-abstract}}
\end{figure}

\smallskip\noindent{\em Structure of the game.}
The game $G_n$ which requires memory of non-elementary size is sketched in 
\figurename~\ref{fig:non-elementary-game} for $n=3$. Its abstract structure is 
shown in \figurename~\ref{fig:non-elementary-game-abstract}, corresponding
to the structure of the counter system in \figurename~\ref{fig:non-elementary-counters}.
The alphabet of player~$1$ is $\{a,b,\#\}$. For the sake of clarity, some
transitions are not depicted in \figurename~\ref{fig:non-elementary-game}. 
It is assumed that for player~$1$, playing an action from a state where this
action has no transition depicted leads to the initial state of the game. 
For example, playing $\#$ in state $q_4$ goes to the initial state, and
from the target state $q_{\smiley}$, all transitions go to the initial state.

\figurename~\ref{fig:non-elementary-game} shows the initial state $q_I$
of the game from which a uniform probabilistic transition branches to the three
states $q_7,r_7,s_7$. The idea of this game is that player~$1$ needs
to ensure that the states $q_1,r_1,s_1$ are reached with positive probability,
so as to ensure that no matter the action ($a$, $b$, or $c$) chosen by player~$2$,
the state $q_{\smiley}$ is reached  with positive probability. From $q_1,r_1,s_1$,
the other actions of player~$2$ (i.e., $b$ and $c$ from $q_1$, $a$ and $c$ from $r_1$, etc.)
lead to the initial state. Player~$2$ can observe the initial state. All the other
states are indistinguishable.

Intuitively, each ``\emph{line}'' of states ($q$'s, $r$'s, and $s$'s) simulate
one counter. Synchronization of the operations on the three counters is ensured by the special 
(and visible to player~$2$) symbol $\#$. Intuitively, since $\#$ is visible to player~$2$,
player~$1$ must play $\#$ at the same ``time'' in the three lines of states
(i.e., after the same number of steps in each line). Otherwise, player~$2$ 
may eliminate one line of states from his belief.
For example, if player~$1$ plays $\#$ in the first step in lines $q$ and $r$,
but not in line $s$, then player~$2$ observing $\#$ can safely update his belief
to $\{q_{\cdot},r_{\cdot}\}$, and thus avoid to play $c$ when one of the states
$q_{1}$, $r_{1}$ is reached. In \figurename~\ref{fig:non-elementary-game}, 
the dotted lines and the subscripts on $\#$ emphasize the layered  
structure of the game, corresponding to the structure of \figurename~\ref{fig:non-elementary-game-abstract}.

From all the above, it follows that player~$1$ needs memory of size non-elementary 
in order to ensure indistinguishable paths ending up in each of the states $q_1,r_1,s_1$,
and win with positive probability. Since all other paths are going back to the
initial state, this strategy can be repeated over and over again to achieve
almost-sure reachability as well.
\ee
\end{example}

\begin{figure}[tbp]
\hrule
\begin{center}
\def\fsize{\normalsize}

\begin{picture}(120,177)(0,0)

{\fsize



\node[Nmarks=r](target)(55,2){$\smiley$}

\rpnode[Nmarks=i](init)(55,172)(4,4){$q_I$}

\node[Nmarks=n](q7)(15,152){$q_7$}

\node[Nmarks=n](q6)(15,132){$q_6$}
\rpnode[Nmarks=n](q5)(15,112)(4,4){$q_5$}

\node[Nmarks=n](ql)(5,97){$q_L$}
\node[Nmarks=n](qr)(25,97){$q_R$}
\node[Nmarks=n](q4)(15,82){$q_4$}
\node[Nmarks=n](q3)(15,62){$q_3$}
\node[Nmarks=n](q2)(15,42){$q_2$}
\node[Nframe=n, Nh=4, Nw=6](q2bis)(37,42){$q_I$}
\node[Nmarks=n](q1)(15,22){$q_1$}

\drawloop[ELpos=50, ELside=l, ELdist=0.5, ELside=l,loopCW=y, loopdiam=6, loopangle=90](q7)
{$\Sigma\setminus \{\#_1\}, -$}

\drawedge[ELpos=50, ELside=l, ELdist=1, curvedepth=0](q7,q6){$\#_1,-$}
\drawedge[ELpos=50, ELside=l, ELdist=1, curvedepth=0](q6,q5){$\#_2,-$}

\drawloop[ELpos=50, ELside=l, ELdist=0.5, ELside=l,loopCW=y, loopdiam=6, loopangle=0](q6)
{$\Sigma\setminus \{\#_2\}, -$}

\drawedge[ELpos=50, ELside=r, ELdist=1, curvedepth=0](q5,ql){$\frac{1}{2}$}
\drawedge[ELpos=50, ELside=l, ELdist=1, curvedepth=0](q5,qr){$\frac{1}{2}$}

\drawedge[ELpos=50, ELside=l, ELdist=.5, curvedepth=0](ql,q4){*}
\drawedge[ELpos=50, ELside=r, ELdist=.5, curvedepth=0](qr,q4){*}

\node[Nframe=n, Nw=0,Nh=0](tmp1)(31,84){}
\node[Nframe=n, Nw=0,Nh=0](tmp2)(31,110){}
\drawline[AHnb=0, arcradius=2](19,82)(31,82)(31,85)
\drawedge[AHnb=0, ELpos=12, ELside=l, ELdist=0, curvedepth=0](tmp1,tmp2){$\begin{array}{l}a,a\\b,b\end{array}$}
\drawline[AHnb=1, arcradius=2](31,109)(31,112)(19,112)

\drawloop[ELpos=50, ELside=l, ELdist=0.5, ELside=l,loopCW=y, loopdiam=6, loopangle=0](q3)
{$\Sigma\setminus \{\#_3\}, -$}

\drawedge[ELpos=50, ELside=l, ELdist=1, curvedepth=0](q4,q3){$\begin{array}{l}a,b\\b,a\end{array}$}
\drawedge[ELpos=50, ELside=l, ELdist=1, curvedepth=0](q3,q2){$\#_3,-$}
\drawedge[ELpos=50, ELside=l, ELdist=1, curvedepth=0](q2,q1){$\begin{array}{l}a,b\\b,a\end{array}$}
\drawedge[ELpos=50, ELside=l, ELdist=1, curvedepth=0](q2,q2bis){$a,a$}
\drawedge[ELpos=50, ELside=r, ELdist=1, curvedepth=0](q2,q2bis){$b,b$}

\drawedge[ELpos=50, ELside=r, ELdist=1, curvedepth=0](q1,target){$-,a$}


\node[Nmarks=n](r7)(55,152){$r_7$}
\rpnode[Nmarks=n](r6)(55,132)(4,4){$r_6$}

\node[Nmarks=n](rl)(45,117){$r_L$}
\node[Nmarks=n](rr)(65,117){$r_R$}
\node[Nmarks=n](r5)(55,102){$r_5$}
\node[Nmarks=n](r4)(55,82){$r_4$}
\node[Nmarks=n](r3)(55,62){$r_3$}
\node[Nframe=n, Nh=4, Nw=6](r3bis)(77,62){$q_I$}
\node[Nmarks=n](r2)(55,42){$r_2$}
\node[Nframe=n, Nh=4, Nw=6](r2bis)(77,42){$q_I$}
\node[Nmarks=n](r1)(55,22){$r_1$}

\drawloop[ELpos=50, ELside=l, ELdist=0.5, ELside=l,loopCW=y, loopdiam=6, loopangle=0](r7)
{$\Sigma\setminus \{\#_1\}, -$}

\drawedge[ELpos=50, ELside=l, ELdist=1, curvedepth=0](r7,r6){$\#_1,-$}

\drawloop[ELpos=50, ELside=l, ELdist=0.5, ELside=l,loopCW=y, loopdiam=6, loopangle=0](r4)
{$\Sigma\setminus \{\#_2\}, -$}

\drawedge[ELpos=50, ELside=r, ELdist=1, curvedepth=0](r6,rl){$\frac{1}{2}$}
\drawedge[ELpos=50, ELside=l, ELdist=1, curvedepth=0](r6,rr){$\frac{1}{2}$}

\drawedge[ELpos=50, ELside=l, ELdist=.5, curvedepth=0](rl,r5){*}
\drawedge[ELpos=50, ELside=r, ELdist=.5, curvedepth=0](rr,r5){*}

\node[Nframe=n, Nw=0,Nh=0](tmp1)(71,104){}
\node[Nframe=n, Nw=0,Nh=0](tmp2)(71,130){}
\drawline[AHnb=0, arcradius=2](59,102)(71,102)(71,105)
\drawedge[AHnb=0, ELpos=12, ELside=l, ELdist=0, curvedepth=0](tmp1,tmp2){$\begin{array}{l}a,a\\b,b\end{array}$}
\drawline[AHnb=1, arcradius=2](71,129)(71,132)(59,132)

\drawloop[ELpos=50, ELside=l, ELdist=.0, loopCW=y, loopdiam=6, loopangle=140](r3){$\begin{array}{l}a,b\\b,a\end{array}$}

\drawedge[ELpos=50, ELside=l, ELdist=1, curvedepth=0](r5,r4){$\begin{array}{l}a,b\\b,a\end{array}$}
\drawedge[ELpos=50, ELside=l, ELdist=1, curvedepth=0](r4,r3){$\#_2,-$}
\drawedge[ELpos=50, ELside=l, ELdist=1, curvedepth=0](r3,r2){$\#_3,-$}
\drawedge[ELpos=50, ELside=l, ELdist=1, curvedepth=0](r2,r1){$\begin{array}{l}a,b\\b,a\end{array}$}

\drawedge[ELpos=50, ELside=l, ELdist=1, curvedepth=0](r3,r3bis){$a,a$}
\drawedge[ELpos=50, ELside=r, ELdist=1, curvedepth=0](r3,r3bis){$b,b$}
\drawedge[ELpos=50, ELside=l, ELdist=1, curvedepth=0](r2,r2bis){$a,a$}
\drawedge[ELpos=50, ELside=r, ELdist=1, curvedepth=0](r2,r2bis){$b,b$}

\drawedge[ELpos=50, ELside=r, ELdist=1, curvedepth=0](r1,target){$-,b$}

\rpnode[Nmarks=i, ilength=3, iangle=90](s7)(95,152)(4,4){$s_7$}

\node[Nmarks=n](sl)(85,137){$s_L$}
\node[Nmarks=n](sr)(105,137){$s_R$}
\node[Nmarks=n](s6)(95,122){$s_6$}
\node[Nmarks=n](s5)(95,102){$s_5$}
\node[Nmarks=n](s4)(95,82){$s_4$}
\node[Nframe=n, Nh=4, Nw=6](s4bis)(117,82){$q_I$}
\node[Nmarks=n](s3)(95,62){$s_3$}
\node[Nframe=n, Nh=4, Nw=6](s3bis)(117,62){$q_I$}
\node[Nmarks=n](s2)(95,42){$s_2$}
\node[Nframe=n, Nh=4, Nw=6](s2bis)(117,42){$q_I$}
\node[Nmarks=n](s1)(95,22){$s_1$}

\drawedge[ELpos=50, ELside=r, ELdist=1, curvedepth=0](s7,sl){$\frac{1}{2}$}
\drawedge[ELpos=50, ELside=l, ELdist=1, curvedepth=0](s7,sr){$\frac{1}{2}$}

\drawedge[ELpos=50, ELside=l, ELdist=.5, curvedepth=0](sl,s6){*}
\drawedge[ELpos=50, ELside=r, ELdist=.5, curvedepth=0](sr,s6){*}

\node[Nframe=n, Nw=0,Nh=0](tmp1)(111,124){}
\node[Nframe=n, Nw=0,Nh=0](tmp2)(111,150){}
\drawline[AHnb=0, arcradius=2](99,122)(111,122)(111,125)
\drawedge[AHnb=0, ELpos=12, ELside=l, ELdist=0, curvedepth=0](tmp1,tmp2){$\begin{array}{l}a,a\\b,b\end{array}$}
\drawline[AHnb=1, arcradius=2](111,149)(111,152)(99,152)

\drawloop[ELpos=50, ELside=l, ELdist=0, loopCW=y, loopdiam=6, loopangle=140](s4){$\begin{array}{l}a,b\\b,a\end{array}$}
\drawloop[ELpos=50, ELside=l, ELdist=0, loopCW=y, loopdiam=6, loopangle=140](s3){$\begin{array}{l}a,b\\b,a\end{array}$}

\drawedge[ELpos=50, ELside=l, ELdist=1, curvedepth=0](s6,s5){$\begin{array}{l}a,b\\b,a\end{array}$}
\drawedge[ELpos=50, ELside=l, ELdist=1, curvedepth=0](s5,s4){$\#_1,-$}
\drawedge[ELpos=50, ELside=l, ELdist=1, curvedepth=0](s4,s3){$\#_2,-$}
\drawedge[ELpos=50, ELside=l, ELdist=1, curvedepth=0](s3,s2){$\#_3,-$}
\drawedge[ELpos=50, ELside=l, ELdist=1, curvedepth=0](s2,s1){$\begin{array}{l}a,b\\b,a\end{array}$}

\drawedge[ELpos=50, ELside=l, ELdist=1, curvedepth=0](s4,s4bis){$a,a$}
\drawedge[ELpos=50, ELside=r, ELdist=1, curvedepth=0](s4,s4bis){$b,b$}
\drawedge[ELpos=50, ELside=l, ELdist=1, curvedepth=0](s3,s3bis){$a,a$}
\drawedge[ELpos=50, ELside=r, ELdist=1, curvedepth=0](s3,s3bis){$b,b$}
\drawedge[ELpos=50, ELside=l, ELdist=1, curvedepth=0](s2,s2bis){$a,a$}
\drawedge[ELpos=50, ELside=r, ELdist=1, curvedepth=0](s2,s2bis){$b,b$}

\drawedge[ELpos=50, ELside=l, ELdist=1, curvedepth=0](s1,target){$-,c$}
\drawedge[ELpos=50, ELside=r, ELdist=1, curvedepth=0](init,q7){$\frac{1}{3}$}
\drawedge[ELpos=50, ELside=l, ELdist=1, curvedepth=0](init,r7){$\frac{1}{3}$}
\drawedge[ELpos=50, ELside=l, ELdist=1, curvedepth=0](init,s7){$\frac{1}{3}$}

\gasset{dash={.15 4}0}

\drawline[AHnb=0, arcradius=2](27,142)(54,142)
\drawline[AHnb=0, arcradius=2](65,142)(75,142)(75,112)(94,93)
\drawline[AHnb=0, arcradius=2](105,92)(117,92)

\drawline[AHnb=0, arcradius=2](27,122)(35,122)(35,92)(54,73)
\drawline[AHnb=0, arcradius=2](65,72)(94,72)
\drawline[AHnb=0, arcradius=2](105,72)(117,72)

\drawline[AHnb=0, arcradius=2](25,52)(54,52)
\drawline[AHnb=0, arcradius=2](65,52)(94,52)
\drawline[AHnb=0, arcradius=2](105,52)(117,52)



}
\end{picture}
\end{center}
\hrule
\caption{{\bf Memory of non-elementary size may be necessary for positive and almost-sure reachability.}
A family of one-sided reachability games in which 
player~$1$ is has perfect observation. Player~$1$ needs memory of non-elementary size
to win positive reachability (as well as almost-sure reachability).
\label{fig:non-elementary-game}}
\end{figure}
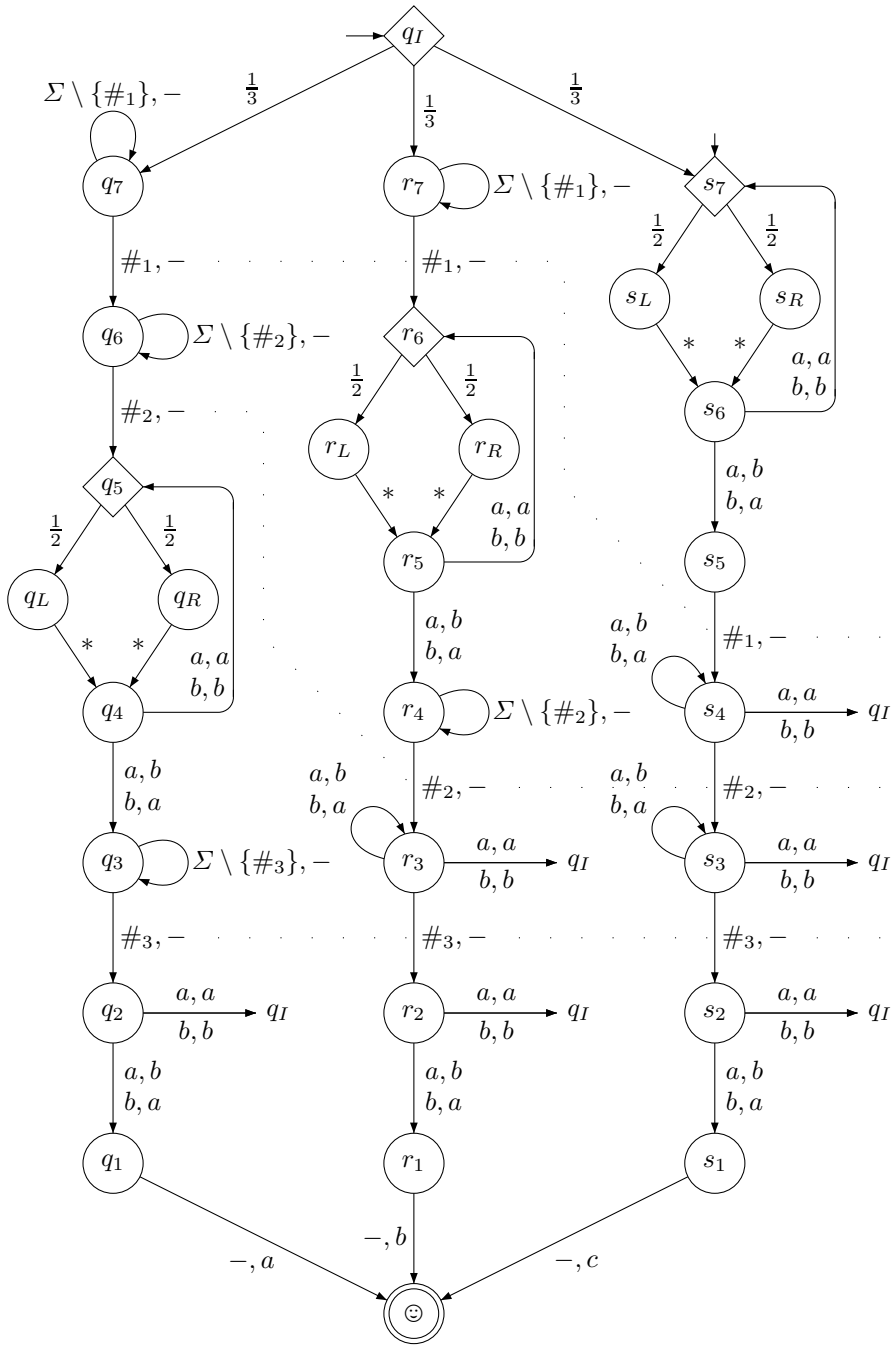

\begin{theorem}\label{theo:non-elementary}
In one-sided partial-observation stochastic games with player~1 perfect
and player~2 partial, both pure almost-sure and pure positive winning strategies for 
reachability objectives for player~1 require memory of 
non-elementary size in general.
\end{theorem}

\subsection{Upper bound for positive reachability with almost-sure safety}\label{sec:positive+sure}


We present the 
solution of one-sided games with a conjunction of positive reachability 
and almost-sure safety objectives, 
in which player~$1$ has perfect observation and player~$2$ has partial observation. 
This will be useful in Section~\ref{sec:as-reach} to solve almost-sure reachability,
and using a trivial safety objective (safety for the whole state space) 
it also gives the solution for positive reachability.

Let $G=\tuple{Q, q_0,\trans_G}$ be a game
over alphabets $A_1,A_2$ and observation set $\Obs_2$ for player~$2$, 
with reachability objective $\Reach(\target)$ (where $\target \subseteq Q$)
and safety objective $\Safe(Q_G)$ (where $Q_G \subseteq Q$ represents a set of good states)
for player~1. 
We assume that the states in $\target$ are absorbing and that $\target \subseteq Q_G$. 
This assumption is satisfied by the games we consider in Section~\ref{sec:as-reach},
as well as by the case of a trivial safety objective ($Q_G = Q$).
The goal of player~1 is to ensure positive probability to reach $\target$ and 
almost-sure safety for the set $Q_G$.

Before presenting the algorithm for solving these games in pure strategies,
we consider the case of randomized strategies. After, we use the results
of randomized strategies to solve the case of pure strategies. 

\paragraph{\bf Step 1 - Winning with randomized strategies.} 
First, we show that with randomized strategies, memoryless strategies are
sufficient. It suffices to play uniformly at random the set of safe actions.
In a state $q$, an action $a \in A_1$ is \emph{safe} if $\Post_G(q,a,b) \subseteq \Win_{\mathit{safe}}$
for all $b \in A_2$, where $\Win_{\mathit{safe}}$ is the set of states that are 
sure winning\footnote{Note that for safety objectives, the notion of sure winning 
and almost-sure winning coincide, and pure strategies are sufficient.} 
for player~$1$ in $G$ for the safety objective $\Safe(Q_G)$. 
This strategy ensures that the set $Q \setminus Q_G$ of bad states is never reached,
and from the positive winning region of player~$1$ for $\Reach(\target)$
it ensures that the set $\target$ is reached 
with positive probability. Therefore, computing the set $Z$ of states 
that are winning for player~$1$ in randomized strategies can be done
by fixing the uniformly randomized safe strategy for player~$1$, and 
checking that player~$2$ does not almost-surely win the safety objective 
$\Safe(Q \setminus \target)$, which requires the analysis of a POMDP for 
almost-sure safety and can be done in exponential time using a simple subset 
construction~\cite[Theorem~2]{CDH10a}.

Note that $\target \subseteq Z$ and that from all states in $Z$, 
player~$1$ can ensure that $\target$ is reached with positive probability 
within at most $2^{\abs{Q}}$ steps,
while from any state $q\not\in Z$, player~$1$ cannot win positively with
a randomized strategy, and therefore also not with a pure strategy.

\paragraph{\bf Step 2 - Pure strategies to simulate randomized strategies.} 
Second, we show that pure strategies can in some cases simulate the
behavior of randomized strategies. As we have seen in the gadget $\inc$
of \figurename~\ref{fig:gadget-non-elementary}, if there are two play
prefixes ending up in the same state and that are indistinguishable for player~$2$ 
(e.g., $q_0 L q_{ab}$ and $q_0 R q_{ab}$ in the example), then player~$1$ can simulate 
a random choice of action over support $\{a,b\}$ by playing $a$ after
$q_0 L q_{ab}$, and playing $b$ after $q_0 R q_{ab}$. No matter the 
choice of player~$2$, one of the plays will reach $q_0$ and the other will
reach the exit state of the gadget. Intuitively, this corresponds to a
uniform probabilistic choice of the actions $a$ and $b$: 
the state $q_0$ and the exit state are reached with probability $\frac{1}{2}$.

In general, if there are $\abs{A_1}$ indistinguishable play prefixes ending up
in the same state $q$, then player~$1$ can simulate a random choice of actions
over $A_1$ from $q$. However, the number of indistinguishable play prefixes 
in a successor state $q'$ may have decreased by a factor $\abs{A_1}$ (there may be
just one play reaching $q$'). Hence, in order to simulate a randomized strategy
during $k$ steps, player~$1$ needs to have $\abs{A_1}^k$ indistinguishable play prefixes.
Since $2^{\abs{Q}}$ steps are sufficient for a randomized strategy
to achieve the reachability objective, an upper bound on the number of play prefixes that 
are needed to simulate a randomized strategy using a pure strategy is $\Num = \abs{A_1}^{2^{\abs{Q}}}$.
More precisely, if the belief of player~$2$ is $B \subseteq Z$ and 
in each state $q \in B$ there are at least $\Num$ indistinguishable play prefixes, then
player~$1$ wins with a pure strategy that essentially simulates a winning
randomized strategy (which exists since $q \in Z$) for $2^n$ steps.

\paragraph{\bf Step 3 - Counting abstraction for pure strategies.} 
We present a construction of a game of perfect observation $H$ such that
player~$1$ wins in $H$ if and only if player~$1$ wins in $G$. The objective
in $H$ is a conjunction of positive reachability and almost-sure safety objectives,
for which pure memoryless winning strategies exist: for every state we restrict the
set of actions to safe actions, and then we solve positive reachability on 
a perfect-observation game. The result follows since for perfect-observation games pure memoryless
positive winning strategies exist for reachability objectives~\cite{Con92}.

\paragraph{State space.} 
The idea of this construction is to keep track of the belief set $B \subseteq Q$ of player~$2$,
and for each state $q \in B$, of the number of indistinguishable play prefixes that end up in $q$.
For $k \in \nat$, we denote by $[k]$ the set $\{0,1,\dots,k\}$.
A state of $H$ is a \emph{counting function} $f:Q \to [\K_*] \cup \{\omega\}$
where $\K_* \in \nat$ is of order $\abs{A_1}^{\abs{A_1}^{\cdot^{\cdot^{\abs{A_1}^{2^{O(n)}}}}}}$
where the number of nested exponentials is in $O(n)$ (where $n = \abs{Q}$).

As we have seen in the example of \figurename~\ref{fig:non-elementary-game}, 
it may be necessary to keep track of a non-elementary number of play prefixes.
We show that the bound $\K_*$ is sufficient, and that we can substitute 
larger numbers by the special symbol $\omega$ to obtain a \emph{finite} counting abstraction.
The belief associated with a counting function $f$ is the set $\Supp(f) = \{q \in Q \mid f(q) \neq 0\}$,
and the states $q$ such that $f(q) = \omega$ are called \emph{$\omega$-states}.

\paragraph{Action alphabet.} In $H$, an action of player~$1$ is a function 
$\hat{a}: Q \times [\K_*] \to~A_1$
that assigns to each copy of a state in the current belief (of player~$2$), 
the action played by player~$1$ after the corresponding play prefix in $G$.
We denote by $\Supp(\hat{a}(q,\cdot)) = \{ \hat{a}(q,i) \mid i \in [\K_*] \}$
the set of actions played by $\hat{a}$ in $q \in Q$.

The action set of player~$2$ in the game $H$ is the same as in $G$. 

\paragraph{Transitions.} 
Let ${\bf 1}(a,A)$ be $1$ if $a \in A$, and $0$ if $a \not\in A$. We denote this function by ${\bf 1}(a \in A)$.
Given $f$ and $\hat{a}$ as above, given an action $b \in A_2$ and an observation $\gamma \in \Obs_2$,
let $f' = \Succ(f,\hat{a},b, \gamma)$ be the function such that $f'(q') = 0$ for all $q' \not\in \gamma$,
and such that for all $q' \in \gamma$:
$$ 
\begin{array}{l}
f'(q') =  \left\{ 
	\begin{array}{ll}
		\omega & \text{ if } \exists a \in \Supp(\hat{a}(q,\cdot)) \cdot \exists q \in Q: f(q) = \omega \land q' \in \Post_G(q,a,b) \\
		x & \text{ otherwise } \\ 
	\end{array} 
\right.  \\
\vspace{-5pt}\\
\text{where } x = \sum_{q \in \Supp(f)} \sum_{i=0}^{f(q)-1} {\bf 1}(q' \in \Post_G(q,\hat{a}(q,i),b)).
\end{array} 
$$
Note that if the current state $q$ is an $\omega$-state, 
then only the support $\Supp(\hat{a}(q,\cdot))$ of the function $\hat{a}$ matters. 

Now $f' = \Succ(f,\hat{a},b, \gamma)$ may not be a counting function
because it may assign values greater than $\K_*$ to some states. 
We show that beyond certain bounds, it is not necessary to remember the 
exact value of the counters and we can replace such large values by $\omega$.
Intuitively, the $\omega$ value can be interpreted as ``very large and definitely positive value".
This abstraction needs to be done carefully in order to obtain the desired 
upper bound (namely, $\K_*$). When a counter $f(q)$ has value $\omega$, 
the successors of $q$ have value $\omega$ according to $\Succ(\cdot)$, which
is faithful if the exact value of the counter $f(q)$ is large enough. 
In fact, large enough means that the counter has value at least $\abs{A_1}$ 
as this allows player~$1$ to play each action at least once. Hence 
the abstraction remains faithful during $\K$ steps if the counters with 
value greater than $\abs{A_1}^{\K}$ are set to $\omega$. We know that if 
all counters have value greater than $\K_1 = \abs{A_1}^{2^{n}}$, then player~$1$
wins by simulating a randomized strategy. Therefore, when all counters but one
have already value $\omega$, we set the last counter to $\omega$ if it has value greater
than $\K_1$. Since this can take at most $\K_1$ steps, the other counters with
value $\omega$ need to have value at least $\K_2 = \K_1 \cdot \abs{A_1}^{\K_1}$. 

Therefore, when all counters but two
have already value $\omega$, whenever a counter gets value greater than $\K_2$ 
we set it to $\omega$. This can take at most $(\K_2)^2$ steps and the other counters with
value $\omega$ need to have value at least $\K_3 = \K_2 \cdot \abs{A_1}^{(\K_2)^2}$.
In general, when all counters but $k$ have value $\omega$, we set a counter to $\omega$
if it has value at least $\K_{k+1} = \K_{k} \cdot \abs{A_1}^{(\K_{k})^{k}}$.
It can be shown by induction that $\K_{k}$ is of order $\abs{A_1}^{\abs{A_1}^{\cdot^{\cdot^{\abs{A_1}^{2^{O(n)}}}}}}$
where the tower of exponential is of height $k$, and thus we do not need
to store counter values greater than $\K_*$. We define the abstraction mapping
$f' = \Abs(f)$ for $f:Q \to \nat$ as follows:

\begin{verse}
\noindent Let $k = \abs{\{q \mid f(q) = \omega\}}$ be the number of counters with
value $\omega$ in $f$. If there is a state $\hat{q}$ with finite value $f(\hat{q})$ greater 
than $\K_{n-k}$, then $f'(\hat{q}) = \omega$ and $f'$ agrees with $f$ on all states 
except $\hat{q}$ (i.e., $f'(q) = f(q)$ for all $q \neq \hat{q}$). Otherwise, $f' = f$.
\end{verse}

Actually, we define $\Abs(f)$ as the $n$th iterate of the above procedure.
Given $f$, $\hat{a}$, and $b$, let $\trans_H(f,\hat{a},b)$ be the uniform
distribution over the set of counting functions $f'$ such that
there exists an observation $\gamma \in \Obs_2$ such that $f' = \Abs(\Succ(f,\hat{a},b, \gamma))$
and $\Supp(f') \neq \emptyset$.

Note that the operators $\Succ(\cdot)$ and $\Abs(\cdot)$ are \emph{monotone},
that is $f \leq f'$ implies $\Abs(f) \leq \Abs(f')$ as well as
$\Succ(f,\hat{a},b, \gamma) \leq \Succ(f',\hat{a},b, \gamma)$ for all $\hat{a},b, \gamma$
(where $\leq$ is the componentwise order).

\paragraph{Objective.} 
Given $\target \subseteq Q$ and $Q_G \subseteq Q$ defining the reachability 
and safety objectives in~$G$, the objective in the game $H$ is a conjunction of positive reachability 
and almost-sure safety objectives, defined by $\Reach(\target_H)$ where\footnote{Recall that $Z$ is the set 
of states that are winning in~$G$ for player~$1$ in randomized strategies.} 
$\target_H = \{f \mid \Supp(f) \subseteq Z \land \forall q \in \Supp(f): f(q) = \omega\}
\cup \{f \mid \Supp(f) \cap \target \neq \emptyset\}$
and by $\Safe(\Good_H)$ where $\Good_H = \{f \mid \Supp(f) \subseteq Q_G\}$.

\paragraph{\bf Step 4 - Correctness argument.} 
First, assume that there exists a pure winning strategy $\straa$ for 
player~$1$ in $G$, and we show how to construct a winning strategy $\straa^H$ in $H$.
As we play the game in $G$ using $\straa$, we keep track of the exact number of indistinguishable
play prefixes ending up in each state. 
This allows to define the action $\hat{a}$ to play in $H$ by collecting the actions played by $\straa$
in all the indistinguishable play prefixes. Note that by monotonicity, the counting
abstractions in the corresponding play prefix of $H$ are at least as big (assuming $\omega > k$ for all $k\in \nat$), 
and thus the action $\hat{a}$ is well-defined. Since $\straa$ is winning,
$\target$ is reached with positive probability in $G$, and the set $Q \setminus Q_G$ is never hit,
and therefore a counting function $f \in \target_H$ (such that $\Supp(f) \cap \target \neq \emptyset$)
is reached with positive probability in $H$, and all plays remain safe in the set $\Good_H$.

Second, assume that there exists a winning strategy $\straa^H$ for 
player~$1$ in $H$, and we show how to construct a pure winning strategy $\straa$ in $G$.
We can assume that $\straa^H$ is pure memoryless.
Fix an arbitrary strategy $\strab$ for player~$2$ and consider the unfolding tree of the game $H$ 
when $\straa^H$ and $\strab$ are fixed (we get a tree and not just a path because the game is stochastic).
In this tree, there is a shortest path to reach $\target_H$ and this path has no loop 
since strategy $\straa^H$ is memoryless.
we show that the length of this path can be bounded, and that the bounds used in the
counting abstraction with $\omega$'s are faithful, showing that the strategy $\straa^H$ 
can be simulated in $G$ (in particular, we need to show that there are sufficiently many 
indistinguishable play prefixes in $G$ to simulate the action `functions' $\hat{a}$ played by $\straa^H$). 
More precisely, the bounds $\K_1, \K_2, \dots$ have been chosen in such a way
that counters with value $\omega$ keep a positive value until all counters
get value $\omega$. For example, when all counters but $k$ have value $\omega$,
it takes at most $(K_k)^k$ steps to get one more counter with value $\omega$
by the argument given in Step 3. Therefore, along the shortest path to $\target_H$,
either we reach a counting function $f$ with $f(q) = \omega$ for all $q \in \Supp(f)$,
or a counting function $f$ with $\Supp(f) \cap \target \neq \emptyset$. In the first case,
we can simulate $\straa^H$ in $G$ to this point, and then win by simulating a
winning randomized strategy, and in the second case the reachability objective  
$\Reach(\target)$ is achieved in $G$ with positive probability. Since the strategy $\straa^H$
ensures that the support of the counting functions never hit the set $Q \setminus Q_G$,
player~$1$ wins in $G$ for the positive reachability and almost-sure safety objectives.

\begin{theorem}\label{theo:complexity-one-sided-player-two}
In one-sided partial-observation stochastic games with player~1 perfect
and player~2 partial,
non-elementary size memory is sufficient for pure strategies to ensure positive 
probability reachability along with almost-sure safety for player~1; 
and hence for pure positive winning strategies for reachability objectives 
for player~1 non-elementary memory bound is optimal.
\end{theorem}

\subsection{Upper bound for almost-sure reachability}\label{sec:as-reach}
In this section we present the algorithm to solve the almost-sure 
reachability problem. 
We start with an example to illustrate that in general  strategies 
for almost-sure winning may be more complicated than positive winning
for reachability objectives.

\begin{figure}[!tb]
\hrule
\begin{center}
\def\fsize{\normalsize}

\begin{picture}(85,57)(0,-1)

{\fsize

\node[Nmarks=i, iangle=135](q1)(20,42){$q_1$}\nodelabel[NLangle=120, ExtNL=y, NLdist=3](q1){\sfrac{1}{2}} 
\node[Nmarks=i, iangle=225](q2)(20,10){$q_2$}\nodelabel[NLangle=240, ExtNL=y, NLdist=3](q2){\sfrac{1}{2}} 
\rpnode[Nmarks=n](r0)(10,26)(4,3.5){}

\node[Nmarks=n, Nmr=0](q3)(45,42){$q_3$}
\node[Nmarks=n, Nmr=0](q4)(45,10){$q_4$}

\rpnode[Nmarks=n](r1)(65,49)(4,3.5){}
\rpnode[Nmarks=n](r2)(65,35)(4,3.5){}
\rpnode[Nmarks=n](r3)(65,17)(4,3.5){}
\rpnode[Nmarks=n](r4)(65,3)(4,3.5){}

\node[Nmarks=r, rdist=0.8](qX)(80,26){\smiley}

\drawedge[ELpos=50, ELside=l, ELdist=1, curvedepth=0](q1,r0){$b,-$}
\drawedge[ELpos=50, ELside=r, ELdist=1, curvedepth=0](q2,r0){$b,-$}


\drawbpedge[ELpos=30, ELside=r, ELdist=.5](r0,180,14,q1,180,18){\sfrac{1}{2}}
\drawbpedge[ELpos=30, ELside=l, ELdist=.5](r0,180,14,q2,180,18){\sfrac{1}{2}}

\drawedge[ELpos=50, ELside=l, ELdist=.5, curvedepth=0](q1,q3){$a,-$}
\drawedge[ELpos=50, ELside=l, ELdist=.5, curvedepth=0](q2,q4){$a,-$}

\drawedge[ELpos=45, ELside=l, ELdist=.5, curvedepth=0](q3,r1){$-,a$}
\drawedge[ELpos=45, ELside=r, ELdist=1, curvedepth=0](q3,r2){$-,b$}
\drawedge[ELpos=45, ELside=l, ELdist=.5, curvedepth=0](q4,r3){$-,a$}
\drawedge[ELpos=45, ELside=r, ELdist=1, curvedepth=0](q4,r4){$-,b$}

\drawbpedge[ELpos=10, ELside=r, ELdist=1, curvedepth=-8](r1,150,18,q1,45,12){\numfrac{1}}
\drawbpedge[ELpos=10, ELside=l, ELdist=1, curvedepth=8](r4,210,18,q2,315,12){\numfrac{1}}
\drawbpedge[ELpos=30, ELside=r, ELdist=.5, curvedepth=8](r2,260,12,q1,332,30){\sfrac{1}{2}}
\drawbpedge[ELpos=40, ELside=l, ELdist=.5, curvedepth=0](r2,280,9,qX,165,8){\sfrac{1}{2}}
\drawbpedge[ELpos=30, ELside=l, ELdist=.5, curvedepth=-8](r3,100,12,q2,28,35){\sfrac{1}{2}}
\drawbpedge[ELpos=40, ELside=r, ELdist=.5, curvedepth=0](r3,80,9,qX,195,8){\sfrac{1}{2}}

\drawloop[ELpos=50, ELside=l, ELdist=1, ELside=l,loopCW=y, loopdiam=7, loopangle=90](qX){}



}
\end{picture}
 
\end{center}
\hrule
\caption{{\bf Almost-sure winning strategy may require more memory than positive winning strategies.}
A one-sided reachability game where player~$1$ (round states) has perfect 
observation, player~$2$ (square states) is blind. Player~$1$ has a pure almost-sure 
winning strategy, but no pure belief-based memoryless strategy is almost-sure winning. 
However, player~$1$ has a pure belief-based memoryless strategy that is positive winning.
\label{fig:belief-not-sufficient}
}

\end{figure}
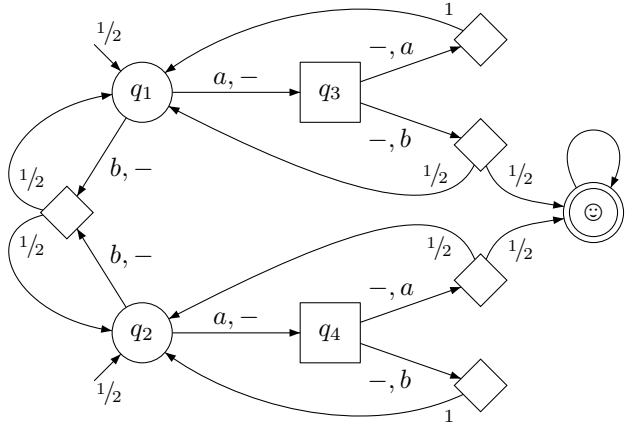

\begin{example}
{\bf Almost-sure winning strategy may require more memory than positive winning strategies.}
The example of \figurename~\ref{fig:belief-not-sufficient} illustrates
a key insight in the algorithmic solution of almost-sure reachability
games where player~$1$
has perfect observation and player~$2$ has partial observation (he is blind in this case).
For player~$1$, playing $a$ in $q_1$ and in $q_2$ is a positive winning
strategy to reach $q_{\smiley}$. This is because from $\{q_1,q_2\}$,
the belief of player~$2$ becomes $\{q_3,q_4\}$ and no matter the action 
chosen by player~$2$, the state $q_{\smiley}$ is reached with positive 
probability from either $q_3$ or $q_4$. 

However, always playing $a$ when the belief of player~$2$ is $\{q_1,q_2\}$
is not almost-sure winning because if player~$2$ chooses always the same 
action (say $a$) in $\{q_3,q_4\}$, then with probability $\frac{1}{2}$ the state 
$q_{\smiley}$ is not reached. Intuitively, this happens because player~$2$
can guess that the initial state is, say $q_1$, and be right with positive 
probability (here $\frac{1}{2}$).
To be almost-surely winning, player~$1$
needs to alternate actions $a$ and $b$ when the belief is $\{q_1,q_2\}$.
The action $b$ corresponds to the \emph{restart phase} of the strategy,
i.e. even assuming that player~$2$'s belief would be, say $\{q_1\}$,
the action $b$ ensures that $q_{\smiley}$ is reached with positive 
probability by make the belief to be $\{q_1,q_2\}$.
\ee
\end{example}

\smallskip\noindent{\em Notation.} We will consider $\target$ as the
set of target states and without loss of generality assume that all
target states are absorbing. 
In this section the belief of player~2 represents the set of states that 
can be with positive probability.
Given strategies $\straa$ and $\strab$ for player~1 and player~2, respectively,
a state $q$ and a set $K \subseteq Q$ we denote by $\Prb_{q,K}^{\straa,\strab}(\cdot)$
the probability measure over sets of paths when the players play the strategies,
the initial state is $q$ and the initial belief for player~2 is $K$.

In rest of this section we omit the subscript $G$ (such as we write $\Strab^O$ 
instead of $\Strab^O_G$) as the game is clear from the context.

\smallskip\noindent{\em Bad states.}
Let $\ov{\target}=Q \setminus \target$. 
Let 
\[
Q_{B} =
\set{ q\in Q \mid \forall \straa \in \Straa^P \cdot \exists \strab \in \Strab^O: \  
\Prb_{q,\{q\}}^{\straa,\strab}(\Safe(\ov{\target})) >0
}
\]
be the set of states $q$ such that given the initial belief of player~2 is 
the singleton $\{q\}$, for all pure strategies for player~1 there is a 
counter observation-based strategy for player~2 to ensure that 
$\Safe(\ov{\target})$ is satisfied with positive probability.
We will consider $Q_B$ as the set of \emph{bad} states.

\smallskip\noindent{\em Property of an almost-sure winning strategy.}
Consider a pure almost-sure winning strategy for player~1 that ensures against 
all observation-based strategies of player~2 that $\target$ is reached with 
probability~1. 
Then we claim that the belief of player~2 must never intersect with $Q_B$: 
otherwise if the belief intersects with $Q_B$, let $q$ be the 
state in $Q_B$ that is reached with positive probability.
Then player~2 simply assumes that the current state is $q$, updates the 
belief to $\{q\}$, and the guess is correct with positive probability. 
Given the belief is $\{q\}$, since $q \in Q_B$, it follows that
against all player~1 pure strategies  there is an observation-based strategy 
for player~2 to ensure with positive probability that $\target$ is not reached.
This contradicts that the strategy for player~1 is almost-sure winning.

\smallskip\noindent{\em Transformation.} We transform the game by changing all
states in $Q_B$ as absorbing. 
Let $Q_G=Q \setminus Q_B$.
By definition we have 
\[
Q_G =
\set{ q\in Q \mid \exists \straa \in \Straa^P \cdot \forall \strab \in \Strab^O: 
\Prb_{q,\{q\}}^{\straa,\strab}(\Reach(\target)) =1 
}.
\]
By the argument above that for a pure almost-sure winning
strategy the belief must never intersect with $Q_B$ we have 
\[
\begin{array}{rcl}
Q_G  & = &
\{\ q\in Q \mid \exists \straa \in \Straa^P\cdot \forall \strab \in \Strab^O: 
\Prb_{q,\{q\}}^{\straa,\strab}(\Reach(\target)) =1 \\
& & \qquad
\text{ and } 
\Prb_{q,\{q\}}^{\straa,\strab}(\Safe(Q\setminus Q_B)) =1 
\ \}.
\end{array}
\]
Let 
\[
\begin{array}{rcl}
Q_G^p & = &
\{\ q\in Q \mid \exists \straa \in \Straa^P\cdot \forall \strab \in \Strab^O: 
\Prb_{q,\{q\}}^{\straa,\strab}(\Reach(\target)) >0 \\
 & & \qquad \text{ and }
\Prb_{q,\{q\}}^{\straa,\strab}(\Safe(Q\setminus Q_B)) =1 
\ \}.
\end{array}
\]
We now show that $Q_G^p=Q_G$. The inclusion $Q_G \subseteq Q_G^p$ is 
trivial, and we now show the other inclusion $Q_G^p \subseteq Q_G$. 
Observe that in $Q_G^p$ we have  the property of positive reachability and 
almost-sure safety and we will use strategies for positive reachability and 
almost-sure safety to construct an almost-sure winning strategy.
We consider $Q_B$ as the set of unsafe states (i.e., $Q_G$ is the safe set), 
and $\target$ as the target and invoke the results of the Section~\ref{sec:positive+sure}: 
for all $q \in Q_G^p$ there is a pure finite-memory strategy  $\straa_q$ of 
memory at most $B$ (where $B$ is non-elementary) to ensure that from $q$, 
within $N=2^{O(B)}$ steps, $\target$ is reached with probability at least some 
positive constant $\eta_q>0$, even when the initial belief for player~2 is 
$\{q\}$. 
Let $\eta=\min_{q \in Q_G^p} \eta_q$.
A pure finite-memory almost-sure winning strategy is described below. 
The strategy plays in two-phases: (1) the \emph{Restart} phase; and 
(1) the \emph{Play} phase.
We define them as follows:
\begin{enumerate}
\item \emph{Restart phase.} Let the current state be $q$, 
assume that the belief for player~2 is $\{q\}$ 
and goto the Play phase with strategy $\straa_q$ that ensures that 
$Q_G$ is never left and $\target$ is reached within $N$ steps with probability 
at least $\eta>0$.
\item \emph{Play phase.} Let $\straa$ be the strategy defined in the 
Restart phase, then play $\straa$ for $N$ steps and go back to the 
Restart phase.
\end{enumerate}
The strategy is almost-sure winning as for all states in $Q_G^p$ and 
for all histories, in every $N$ steps the probability to reach $\target$ 
is at least $\eta>0$, and  $Q_G$ (and hence $Q_G^p$) is never left.
Thus probability to reach $\target$ in $N \cdot \ell$ steps, 
for $\ell \in \Nats$, is at least $1- (1-\eta)^\ell$ and this is~1 as 
$\ell \to \infty$.
Thus the desired result follows and we obtain the almost-sure winning strategy.

\smallskip\noindent{\bf Memory bound and algorithm.} 
The memory upper bound for the almost-sure winning strategy constructed is as follows:
$\abs{Q} \cdot B + \log N$, we require $\abs{Q}$ strategies of Section~\ref{sec:positive+sure} of memory
size $B$ and a counter to count up to $N=2^{O(B)}$ steps.
We now present an algorithm for almost-sure reachability that works in time 
$2^{\abs{Q}} \times O($\PosReachSafe$)$, where \PosReachSafe\ denote the complexity 
to solve the positive reachability along with almost-sure safety problem.
The algorithm enumerates all subset $Q' \subseteq Q$ and 
then verify that forall $q \in Q'$ player~1 can ensure to reach $\target$ with 
positive probability staying safe in $Q'$ with probability~1.
In other words the algorithm enumerates all subsets $Q' \subseteq Q$ to 
obtain the set $Q_G$.
The enumeration is exponential and the verification requires solving the 
positive reachability with almost-sure safety problem.


\begin{theorem}
In one-sided partial-observation stochastic games with player~1 perfect
and player~2 partial,
non-elementary size memory is sufficient for pure strategies to ensure almost-sure 
reachability for player~1; and hence for pure almost-sure winning strategies for 
reachability objectives for player~1 non-elementary memory bound is optimal.
\end{theorem}

\begin{corollary} 
In one-sided partial-observation stochastic games with player~1 perfect
and player~2 partial,
the problem of deciding the existence of pure almost-sure and positive winning
strategies for reachability objectives for player~1 can be solved in 
non-elementary time complexity.  
\end{corollary}

\begin{comment}

\begin{enumerate}

\item EXPTIME upper bound: enumerate all subset $Q' \subseteq Q$ and 
verify that forall $q \in Q'$ player~1 can ensure to reach $T$ with 
positive probability staying safe in $Q \setminus Q'$ with probability~1.
The enumeration is exponential and the verification is exponential 
by previous section.

\item EXPTIME lower bound: lower bound POMDP with almost-safety.

\end{enumerate}

\subsection{Sure safety}
This is PTIME-complete. We can treat player~2 as perfect-information. 
Because if there is a path with positive probability to bad state, 
then player~2 can just play the correct actions even being blind. 
So this is perfect-information safety and we have PTIME-completeness.

\subsection{Positive Safety}

\smallskip\noindent{\bf Main idea.}
The belief based construction: given belief $B$, player~1 chooses action $a$,
then player~2 chooses action~$b$, and we come to $(B,a,b)$. 
If goes out of safe set $F$, then it is a bad state, and belief updated to 
$\Post_{a,b}(B) \cap F$ and observation. Goal to visit bad states finitely 
often.
If bad states visited finitely often, then after some point on good set never 
left, that is no probability to go out of $F$. Hence positive safety.
If bad visited infinitely often, then infinitely often positive probability to
go out, and hence out with probability~1. 
We need to solve exponential size turn-based coB\"uchi game.

\mynote{Open: To be done}

\mynote{Open: in cases above with EXPTIME bound: memory lower bounds.}

\end{comment}

\section{Finite-memory Strategies for Two-sided Games}
In this section we show the existence of finite-memory pure strategies 
for positive and almost-sure winning in two-sided games.

\subsection{Positive reachability with almost-sure safety}
Let $\target$ be the set of target states for reachability 
(such that all the target states are absorbing) and $Q_G$ be the set of good states
for safety with $\target \subseteq Q_G$.
Our goal is to show that for pure strategies to ensure positive probability 
reachability to $\target$ and almost-sure safety for $Q_G$, 
finite-memory strategies suffice.
Note that with $Q_G$ as the whole state space we obtain the result for positive
reachability as a special case.

\begin{lemma}\label{lemm_pos_reach_safe_gen}
For all games $G$, for all $q \in Q$, if there exists a pure 
strategy $\straa \in \Straa^O\cap \Straa^P$ such that for all 
strategies $\strab \in \Strab^O$ of player~2  we have 
\[
\Prb_{q}^{\straa,\strab}(\Reach(\target))>0 \quad \text{ and } \quad
\Prb_{q}^{\straa,\strab}(\Safe(Q_G))=1; 
\] 
then there exists a finite-memory pure 
strategy $\straa^f \in \Straa^O\cap \Straa^P$ such that for all 
strategies $\strab \in \Strab^O$ of player~2 we have 
\[
\Prb_{q}^{\straa^f,\strab}(\Reach(\target))>0 \quad \text{ and } \quad
\Prb_{q}^{\straa^f,\strab}(\Safe(Q_G))=1. 
\] 
\end{lemma}
 
We prove the result with the following two claims. 
We fix a (possibly infinite memory) strategy $\straa \in \Straa^O \cap \Straa^P$
such that for all 
strategies $\strab \in \Strab^O$ of player~2  we have 
\[
\Prb_{q}^{\straa,\strab}(\Reach(\target))>0 \quad \text{ and } \quad
\Prb_{q}^{\straa,\strab}(\Safe(Q_G))=1. 
\] 

\smallskip\noindent{\bf Claim~1.} If there exists $N \in \Nats$ such that 
for all 
strategies $\strab \in \Strab^O$ of player~2  we have 
\[
\Prb_{q}^{\straa,\strab}(\Reach^{\leq N}(\target))>0 \quad \text{ and } \quad
\Prb_{q}^{\straa,\strab}(\Safe(Q_G))=1 
\] 
where $\Reach^{\leq N}$ denotes reachability within first $N$-steps;
then there exists a finite-memory pure 
strategy $\straa^f \in \Straa^O\cap \Straa^P$ such that for all 
strategies $\strab \in \Strab^O$ of player~2 we have 
\[
\Prb_{q}^{\straa^f,\strab}(\Reach(\target))>0 \quad \text{ and } \quad
\Prb_{q}^{\straa^f,\strab}(\Safe(Q_G))=1. 
\] 
\begin{proof}
The finite-memory strategy $\straa^f$ is as follows: play like the 
strategy $\straa$ for the first $N$-steps, and then switch to a strategy
to ensure $\Safe(Q_G)$ with probability~1.
The strategy ensure positive probability reachability to $\target$ as for 
the first $N$-steps it plays like $\straa$ and $\straa$ already ensures
positive reachability within $N$-steps.
Moreover, since $\straa$ ensures $\Safe(Q_G)$ with probability~1, it must also 
ensure $\Safe(Q_G)$ for the first $N$-steps, and since $\straa^f$ after the
first $N$-steps only plays a strategy for almost-sure safety, it follows that 
$\straa^f$ guarantees $\Safe(Q_G)$ with probability~1.
The strategy $\straa^f$ is a finite-memory strategy since it needs to play 
like $\straa$ for the first $N$-steps (which requires finite-memory) and
then it switches to an almost-sure safety strategy for which exponential size
memory is sufficient (for safety objective almost-sure winning coincides
with sure winning and then belief-based strategies are sufficient; 
see~\cite{CD10b} for details).
\qed
\end{proof}

\smallskip\noindent{\bf Claim~2.} There exists $N \in \Nats$ such that 
for all 
strategies $\strab \in \Strab^O$ of player~2  we have 
\[
\Prb_{q}^{\straa,\strab}(\Reach^{\leq N}(\target))>0 \quad \text{ and } \quad
\Prb_{q}^{\straa,\strab}(\Safe(Q_G))=1 
\] 
where $\Reach^{\leq N}$ denotes reachability within first $N$-steps.

\begin{proof}
The proof is by contradiction. Towards contradiction, assume that 
for all $n \in \nat$, there exists a strategy $\strab_n \in \Strab^O$
such that either $\Prb_{q}^{\straa,\strab_n}(\Reach^{\leq n}(\target)) = 0$ or 
$\Prb_{q}^{\straa,\strab_n}(\Safe(Q_G)) < 1$.

If for some $n \geq 0$ we have $\Prb_{q}^{\straa,\strab_n}(\Safe(Q_G)) < 1$,
then we get a contradiction with the fact that $\Prb_{q}^{\straa,\strab}(\Safe(Q_G))=1$
for all $\strab \in \Strab^O$. Hence $\Prb_{q}^{\straa,\strab_n}(\Safe(Q_G)) = 1$
for all $n \in \nat$, and therefore $\Prb_{q}^{\straa,\strab_n}(\Reach^{\leq n}(\target)) = 0$
for all $n \in \nat$. Equivalently, all play prefixes 
of length at most $n$ and compatible with $\straa$ and $\strab_n$ avoid to hit $\target$, and thus 
$\Prb_{q}^{\straa,\strab_n}(\Safe^{\leq n}(Q \setminus \target)) = 1$ for all $n \in \nat$.
Note that we can assume that each strategy $\strab_n$ is pure because once 
the strategy $\straa$ of player~$1$ is fixed we get a POMDP for player~$2$, and 
for POMDPs pure strategies are as powerful as randomized 
strategies~\cite{CDGH10} (in~\cite{CDGH10} the result was shown for finite 
POMDPs with finite action set, but the proof is based on induction on the 
action set and also works for countably infinite POMDPs). 

Using a simple extension of K\"onig's Lemma~\cite{Konig36}, we construct a strategy $\strab' \in \Strab^O$
such that $\Prb_{q}^{\straa,\strab'}(\Safe(Q \setminus \target)) = 1$.
The construction is as follows. In the initial state $q$, there is an action $b_0 \in A_2$
which is played by infinitely many strategies $\strab_n$. We define $\strab'(q) = b_0$
and let $P_0$ be the set $\{\pi_n \mid \pi_n(q) = b_0\}$. Note that $P_0$ is an infinite set.
We complete the construction as follows.
Having defined $\strab'(\rho)$ for all play prefixes $\rho$ of length at most $k$, and given the
infinite set $P_k$, we define $\strab'(\rho')$ for all play prefixes $\rho'$ of length $k+1$
and the infinite set $P_{k+1}$ as follows. Consider the tuple $b_{\strab_n} \in A_2^m$ of actions
played by the strategy $\strab_n \in P_{k}$ after the $m$ prefixes $\rho'$ of length $k+1$.
Clearly, there exists an infinite subset $P_{k+1}$ of $P_{k}$ in which all strategies play
the same tuple $b_{k+1}$. We define $\strab(\rho')$ using the tuple $b_{k+1}$. 
This construction ensures that no play prefix of length $k+1$ compatible with $\straa$ and $\strab'$
hit the set $\target$, since $\strab'$ agrees with some strategy $\pi_n$ for arbitrarily large $n$.
Repeating this inductive argument yields a strategy $\strab'$ such that 
$\Prb_{q}^{\straa,\strab'}(\Safe(Q \setminus \target)) = 1$, in contradiction with the fact 
that $\Prb_{q}^{\straa,\strab}(\Reach(\target))>0$
for all $\strab \in \Strab^O$. Hence, the desired result follows.
\qed
\end{proof}

The above two claims establish Lemma~\ref{lemm_pos_reach_safe_gen} and gives
the following result.

\begin{theorem}
In two-sided partial-observation stochastic games finite memory is sufficient 
for pure strategies to ensure positive probability reachability along with almost-sure safety
for player~1; 
and hence for pure positive winning strategies for reachability objectives 
finite memory is sufficient and non-elementary memory is required in general for player~1.
\end{theorem}

\subsection{Almost-sure reachability}
We now show that for pure 
strategies for almost-sure reachability, finite-memory strategies suffice. 
The proof is a straight forward extension of the results of Section~\ref{sec:as-reach},
and for finite-memory strategies for positive reachability with almost-sure 
safety we use the result of the previous subsection.

\smallskip\noindent{\em Notation.} We will consider $\target$ as the
set of target states and without loss of generality assume that all
target states are absorbing. 
In this section the belief of player~2 represents the set of states that 
can be with positive probability.
Given strategies $\straa$ and $\strab$ for player~1 and player~2, respectively,
a state $q$ and a set $K \subseteq Q$ we denote by $\Prb_{q,K}^{\straa,\strab}(\cdot)$
the probability distribution when the players play the strategies,
the initial state is $q$ and the initial belief for player~2 is $K$.

In rest of this section we omit subscript $G$ (such as we write $\Strab^O$ 
instead of $\Strab^O_G$) as the game is clear from the context.

\smallskip\noindent{\em Bad beliefs.}
Let $\ov{\target}=Q \setminus \target$. 
Let 
\[
Q_{B} =
\set{ \calb\in 2^Q \mid \forall \straa \in \Straa^O\cap\Straa^P \cdot 
\exists \strab \in \Strab^O \cdot \exists q \in \calb: 
\Prb_{q,\{q\}}^{\straa,\strab}(\Safe(\ov{\target})) >0
}
\]
be the set of beliefs $\calb$ such that for all pure strategies for player~1 
there is a counter strategy for player~2 with a state $q \in \calb$ 
to ensure that given the initial belief of player~2 is the singleton $\{q\}$, 
$\Safe(\ov{\target})$ is satisfied with positive probability.
We will consider $Q_B$ as the set of \emph{bad} beliefs.

\smallskip\noindent{\em Property of an almost-sure winning strategy.}
Consider a pure almost-sure winning strategy for player~1 
that ensures against all strategies of player~2 that 
$\target$ is reached with probability~1. 
Then we claim that the belief of player~2 must never intersect with $Q_B$: 
otherwise if the belief intersects with $Q_B$, let $\calb$ be the 
belief in $Q_B$ that is reached with positive probability.
Then there exists $q \in \calb$ such that player~2 can simply assume that the 
current state is $q$, update the belief to $\{q\}$, and the guess is correct 
with positive probability, and then player~2 can ensure  that against all 
player~1 pure strategies there is a strategy for player~2 to ensure with 
positive probability that $\target$ is not reached.
This contradicts that the strategy for player~1 is almost-sure winning.
Let $Q_G=2^Q \setminus Q_B$.
By definition we have 
\[
Q_G =
\set{ \calb \in 2^Q \mid \exists \straa \in \Straa^O\cap \Straa^P \cdot \forall \strab \in \Strab^O\cdot 
\forall q \in \calb:   
\Prb_{q,\{q\}}^{\straa,\strab}(\Reach(\target)) =1 
}.
\]
By the argument above that for a pure almost-sure winning
strategy the belief must never intersect with $Q_B$ we have 
\[
\begin{array}{rcl}
Q_G & = & 
\{\ \calb \in 2^Q \mid \exists \straa \in \Straa^O\cap \Straa^P \cdot \forall \strab \in \Strab^O\cdot 
\forall q \in \calb:  
\Prb_{q,\{q\}}^{\straa,\strab}(\Reach(\target)) =1 \\
 & & \quad \text{ and }
\Prb_{q,\{q\}}^{\straa,\strab}(\Safe(2^Q\setminus Q_B)) =1 
\ \}.
\end{array}
\]
Let 
\[
\begin{array}{rcl}
Q_G^p & = &
\set{ \calb \in 2^Q \mid \exists \straa \in \Straa^O\cap \Straa^P \cdot \forall \strab \in \Strab^O\cdot 
\forall q \in \calb: 
\Prb_{q,\{q\}}^{\straa,\strab}(\Reach(\target)) >0 \\
& & \quad \text{ and }
\Prb_{q,\{q\}}^{\straa,\strab}(\Safe(2^Q\setminus Q_B)) =1 
}.
\end{array}
\]
We now show that $Q_G^p=Q_G$. The inclusion $Q_G \subseteq Q_G^p$ is 
trivial, and we now show the other inclusion $Q_G^p \subseteq Q_G$.
Observe that in $Q_G^p$ we have  the property of positive reachability and 
almost-sure safety and we will use strategies for positive reachability and 
almost-sure safety to construct a witness finite-memory 
almost-sure winning strategy.
Note that here we have safety for a set of beliefs (instead of set of states, 
and it is straight forward to verify that the argument of the previous 
subsection holds when the safe set is a set of beliefs).
We consider $Q_B$ as the set of unsafe beliefs  (i.e., $Q_G$ is the safe set), 
and $\target$ as the target and invoke the results of the previous subsection: 
for all $\calb \in Q_G^p$ there is a pure finite-memory strategy  $\straa_\calb$ 
of to ensure that from all states $q\in\calb$, within $N$ steps (for some
finite $N\in \Nats$), $\target$ is reached with probability at least some 
positive constant $\eta_\calb>0$, even when the initial belief for player~2 is 
$\{q\}$. 
Let $\eta=\min_{\calb \in Q_G^p} \eta_\calb$.
A pure finite-memory almost-sure winning strategy is described below. 
The strategy plays in two-phases: (1) the \emph{Restart} phase; and 
(1) the \emph{Play} phase.
We define them as follows:
\begin{enumerate}
\item \emph{Restart phase.} Let the current belief be $\calb$, 
the belief for player~2 is any perfect belief $\{q\}$, for $q \in \calb$; 
and goto the Play phase with strategy $\straa_\calb$ that ensures that 
$Q_G$ is never left and $\target$ is reached within $N$ steps with probability 
at least $\eta>0$.
\item \emph{Play phase.} Let $\straa$ be the strategy defined in the 
Restart phase, then play $\straa$ for $N$ steps and go back to the 
Restart phase.
\end{enumerate}
The strategy is almost-sure winning as for all states in $Q_G^p$ and 
for all histories, in every $N$ steps the probability to reach $\target$ 
is at least $\eta>0$, and  $Q_G$ (and hence $Q_G^p$) is never left.
Thus probability to reach $\target$ in $N \cdot \ell$ steps, 
for $\ell \in \Nats$, is at least $1- (1-\eta)^\ell$ and this is~1 as 
$\ell \to \infty$.
Thus the desired result follows and we obtain the required finite-memory 
almost-sure winning strategy.

\smallskip\noindent{\bf Memory bound and algorithm.} 
The memory upper bound for the almost-sure winning strategy constructed is as follows:
$\abs{2^Q} \cdot B + \log N$, we require $\abs{2^Q}$ strategies of the previous subsection 
of memory size $B$ and a counter to count up to $N$ steps; 
where $B$ is the memory required for strategies to ensure 
positive reachability with almost-sure safety objectives.

\begin{theorem}
In two-sided partial-observation stochastic games finite memory is sufficient 
(and non-elementary memory is required in general) for pure strategies for 
almost-sure winning for reachability objectives for player~1.
\end{theorem}



\section{Equivalence of Randomized Action-invisible Strategies and Pure Strategies}
In this section, we show that for two-sided partial-observation games, 
the problem of almost-sure winning with randomized action-invisible strategies  
is inter-reducible with the problem of almost-sure winning with pure strategies.
The reductions are polynomial in the number of states in the game (the reduction
from randomized to pure strategies is exponential in the number of actions).

It follows from the reduction of pure to randomized action-invisible strategies 
that the memory lower bounds for pure strategies transfer to randomized
strategies, and in particular belief-based memoryless strategies
are not sufficient, showing that a remark (without proof) of~\cite[p.4]{CDHR07} 
and the result and construction of~\cite[Theorem~1]{GS09} are wrong.


\subsection{Reduction of randomized action-invisible strategies to pure strategies}\label{sec:randomized-2-pure}
We give a reduction for almost-sure winning for 
randomized action-invisible strategies to pure strategies.
Given a stochastic game $G$ we will construct another stochastic game $H$ 
such that there is a randomized action-invisible almost-sure winning strategy in $G$ iff 
there is a pure almost-sure winning strategy in $H$.
We first show in Lemma~\ref{lemm:red1} the correctness of the reduction for 
finite-memory randomized action-invisible strategies, and then show in 
Lemma~\ref{lemm:red2} that finite memory is sufficient in 
two-sided partial-observation games for randomized action-invisible strategies.


\smallskip\noindent{\bf Construction.} Given a stochastic game 
$G=\tuple{Q,q_0,\trans}$ over action sets $A_1$ and $A_2$, and observations 
$\Obs_1$ and $\Obs_2$ (along with the corresponding observation mappings 
$\obs_1$ and $\obs_2$), we construct a game 
$H=\tuple{Q, q_0, \trans_H}$ over 
action sets $2^{A_1}\setminus \{\emptyset \}$ and $A_2$ and observations $\Obs_1$ 
and $\Obs_2$.
The transition function $\trans_H$ is defined as follows: 
\begin{itemize}
\item for all $q \in Q$ and $A \in 2^{A_1}\setminus \{\emptyset \}$ and $b \in A_2$
we have $\trans_H(q,A,b)(q')=\frac{1}{|A|} \cdot \sum_{a \in A} \trans(q,a,b)(q')$,
i.e., in a state in $Q$ player~1 selects a non-empty subset $A\subseteq A_1$ of actions
and the transition function $\trans_H$ simulates the transition function
$\trans$ along with the uniform distribution over the set $A$ of actions.
\end{itemize}
The observation mappings $\obs_i^H$ in $H$, for $i \in \set{1,2}$ are as 
follows: $\obs_i^H(q)=\obs_i(q)$, where $\obs_i$ is the observation 
mapping in $G$.

\begin{lemma}\label{lemm:red1}
The following assertions hold for reachability objectives:
\begin{enumerate}
\item If there is a pure 
almost-sure winning strategy in 
$H$, then there is a randomized 
action-invisible almost-sure winning strategy in $G$.

\item If there is a finite-memory randomized 
action-invisible almost-sure winning strategy in $G$, then there is a pure 
almost-sure winning strategy in $H$.

\end{enumerate}
\end{lemma}
\begin{proof}
We present both parts of the proof below.

\begin{enumerate}
\item Let $\straa_H$ be a pure almost-sure winning strategy in $H$.
We construct a randomized action-invisible almost-sure winning strategy $\straa_G$ in $G$.
The strategy $\straa_G$ is as constructed as follows.
Let $\rho_G=q_0 q_1 \ldots q_k$ be a play prefix in $G$, and we consider the 
same play prefix $\rho_H=q_0 q_1 \ldots q_k$ in $H$, and let $A_k=\straa_H(\rho_H)$. 
The strategy $\straa_G(\rho_G)$ plays all actions in $A_k$ uniformly at 
random.
Since $\straa_H$ is an almost-sure winning strategy it follows 
$\straa_G$ is also almost-sure winning.
Also observe that if $\straa_H$ is observation-based, then so is $\straa_G$.

\item Let $\straa_G$ be a finite-memory randomized action-invisible almost-sure winning 
strategy in $G$. 
If the strategy $\straa_G$ is fixed in $G$ we obtain a finite POMDP, and 
by the results of~\cite{CDH10a} it follows that in an POMDP the precise transition
probabilities do not affect almost-sure winning.
Hence if $\straa_G$ is almost-sure winning, then the uniform version 
$\straa_G^u$ of the strategy $\straa_G$ that always plays the same support of
the probability distribution as $\straa_G$ but plays all actions in the 
support uniformly at random is also almost-sure winning.
Given $\straa_G^u$ we construct a pure almost-sure winning strategy 
$\straa_H$ in $H$.
Given a play prefix $\rho_H=q_0 q_1 \ldots q_k$ in $H$,
consider the same play prefix  $\rho_G=q_0 q_1 \ldots q_k$ in $G$.
Let $A_k=\Supp(\straa_G^u(\rho_G))$, then $\straa_H(\rho_H)$ plays the action 
$A_k \in (2^{A_1}\setminus \{\emptyset \})$.
Since $\straa_G^u$ is almost-sure winning it follows that $\straa_H$ is 
almost-sure winning.
Observe that if $\straa_G$ is observation-based, then so is $\straa_G^u$, and 
then so is $\straa_H$.

\end{enumerate}

\noindent The desired result follows.
\qed
\end{proof}

\begin{lemma}\label{lemm:red2}
For reachability objectives,
if there exists a randomized 
action-invisible almost-sure winning strategy in $G$, then there exists also a 
finite-memory randomized 
action-invisible almost-sure winning strategy in $G$.
\end{lemma}

\begin{proof}
Let $\calw=\set{\calb \mid \calb \in 2^Q \text{ is the belief of player~1 such that } 
\exists \straa \in \Straa^O\cdot \forall \strab \in \Strab^O\cdot \forall q \in \calb:  
\Prb_{q}^{\straa,\strab}(\Reach(\target))=1}$
denote the set of belief sets $\calb$ for player~1 such that player~1 has a 
(possibly infinite-memory) randomized action-invisible almost-sure winning strategy from 
all starting states in $\calb$.
It follows that the almost-sure winning strategy must ensure that the set 
$\calw$ is never left: this is because from the complement set of $\calw$ 
against all randomized action-invisible for player~1 there is a counter 
strategy for player~2 to ensure that with positive probability  
the target is not reached.
Moreover for all $\calb\in \calw$ the almost-sure winning strategy also 
ensures that $\target$ is reached with positive probability.
Hence we have again the problem of positive reachability with almost-sure 
safety.
We simply repeat the proof for the pure strategy case, treating sets of actions 
(that is the support of the randomized strategy) as actions 
(for pure strategy) and played uniformly at random (as in the reduction from 
$G$ to $H$), and thus obtain a witness finite-memory strategy $\straa_G$ to 
ensure positive reachability and almost-sure safety. 
Repeating the strategy $\straa_G$ with play phase and repeat phase (as in the 
case of pure strategies) we obtain the desired finite-memory almost-sure 
winning strategy.
\qed
\end{proof}

The following theorem follows from the previous two lemmas.
\begin{theorem}
Given a two-sided (resp. one-sided) partial-observation stochastic game 
$G$ with a reachability objective we can construct in time polynomial 
in the size of the game and exponential in the size of the action sets a
two-sided (resp. one-sided) partial-observation stochastic game $H$
such that there exists a randomized action-invisible almost-sure winning strategy
in $G$ iff there exists a pure almost-sure winning strategy in $H$.  
\end{theorem}

For positive winning, randomized memoryless strategies are sufficient (both for
action-visible and action-invisible) and the problem is PTIME-complete 
for one-sided and EXPTIME-complete for two-sided~\cite{BGG09}. 
The above theorem along with Theorem~\ref{theo:complexity-one-sided-player-one} 
gives us the following corollary for almost-sure winning for randomized action-invisible
strategies.

\begin{corollary}
Given one-sided partial-observation stochastic games with player~1 partial
and player~2 perfect, the following assertions hold for
reachability objectives for player~1:
\begin{enumerate}
\item \emph{(Memory complexity).} Exponential memory is sufficient 
for randomized action-invisible strategies for almost-sure winning.

\item \emph{(Algorithm).} The existence of a randomized action-invisible 
almost-sure winning strategy can be decided in time exponential in the state 
space of the game and exponential in the size of the action sets.

\item \emph{(Complexity).} The problem of deciding the existence of a 
randomized action-invisible almost-sure winning strategy is EXPTIME-complete.
\end{enumerate}
\end{corollary}

\subsection{Reduction of pure strategies to randomized action-invisible strategies}\label{sec:pure-2-randomized}
We present a reduction for almost-sure winning for pure strategies
to randomized action-invisible strategies.
Given a stochastic game $G$ we construct another stochastic game $H$ 
such that there exists a pure almost-sure winning strategy in $G$ iff 
there exists a randomized almost-sure winning strategy in $H$.

The idea of the reduction is to force player~$1$ to play a pure strategy
in $H$. The game $H$ simulates $G$ and requires player~$1$ to repeat
each actions played (i.e. to play each action two times). Then, if player~$1$ 
uses randomization, he has to repeat the actions chosen randomly in the previous step.
Since the actions are invisible, this can be achieved only if the support 
of the randomized actions is a singleton, i.e., the strategy is pure.
Note that the reduction works for randomized strategies with actions invisible,
and not when the actions are visible.

\smallskip\noindent{\bf Construction.} Given a stochastic game 
$G=\tuple{Q,q_0,\trans_G}$ over action sets $A_1$ and $A_2$, and observations 
$\Obs_1$ and $\Obs_2$ (along with the corresponding observation mappings 
$\obs_1$ and $\obs_2$), we construct a game 
$H=\tuple{Q \cup (Q \times A_1) \cup \{\sink\}, q_0, \trans_H}$ over the same
action sets $A_1$ and $A_2$ and observations $\Obs_1$ and $\Obs_2$.
The transition function $\trans_H$ is defined as follows: 
\begin{itemize}
\item for all $q \in Q$ and $a \in A_1$ and $b \in A_2$
we have $\trans_H(q,a,b)((q,a))=1$, i.e., in a state $q$ for action $a$ of player~$1$,
irrespective of the choice of player~2, the game stores player~$1$'s action with probability~1;
\item for all $(q,a) \in Q \times A_1$, for all $b \in A_2$ we have 
$\trans_H((q,a),a,b) = \trans_G(q,a,b)$, i.e. if player~$1$ repeats the action
played in the previous step, then the probabilistic transition function is the same as in $G$;
and for all $a' \in A_1 \setminus \{a\}$, we have $\trans_H((q,a),a,b)(\sink) = 1$, i.e. 
if player~$1$ does not repeat the same action, then the sink state is reached. 
\item for all $a \in A_1$ and $b \in A_2$, we have $\trans_H(\sink,a,b)(\sink) = 1$.
\end{itemize}
The observation mappings $\obs_i^H$ in $H$ ($i \in \{1,2\}$) are as follows:
$\obs_i^H(q) = \obs_i^H((q,a)) = \obs_i(q)$, where $\obs_i$ is the observation 
mapping in $G$. Note that $H$ is of size polynomial in the size of $G$.

\begin{lemma}\label{lemm:red-rand-pure}
Let $\target \subseteq Q$ be a set of target states. 
There exists a pure almost-sure winning strategy in~$G$ for $\Reach(\target)$
if and only if there exists a 
randomized action-invisible almost-sure winning strategy in $H$ for objective $\Reach(\target)$. 
\end{lemma}

\begin{proof}
We present both directions of the proof below.

\begin{enumerate}
\item Let $\straa_H$ be a randomized action-invisible almost-sure winning strategy in $H$.
We show that we can assume wlog that $\straa_H$ is actually a pure strategy. 
To see this, assume that under strategy $\straa_H$ there is a prefix 
$\rho_H = q_0 (q_0,a_0) q_1 (q_1,a_1) \ldots q_k$ in $H$ compatible with $\straa_H$
from which $\straa_H$ plays a randomized action with support $A \subseteq A_1$ and $\abs{A} > 1$.
Then, with positive probability the states $(q_k, a_k)$ and $(q_k, a'_k)$
are reached where $a_k,a'_k \in A$ and $a_k \neq a'_k$. No matter the action(s)
played by $\straa_H$ in the next step, the state $\sink$ is reached with positive
probability in the next step, either from $(q_k, a_k)$ or from $(q_k, a'_k)$. 
This contradicts that $\straa_H$ is almost-sure winning.
Therefore, we can assume that $\straa_H$ is a pure strategy that repeats each action
two times. We construct a pure almost-sure winning strategy in $G$ by removing these repetitions.

\item Let $\straa_G$ be a pure almost-sure winning strategy in $G$. 
Consider the strategy $\straa_H$ in $H$ that always repeats two times
the actions played by $\straa_G$. The strategy $\straa_H$ is observation-based
and almost-sure winning since $H$ simulates $G$ when actions are repeated twice.
\end{enumerate}

\noindent The desired result follows.
\qed
\end{proof}

\begin{figure}[!tb]
\hrule
\begin{center}
\def\fsize{\normalsize}

\begin{picture}(97,61)(0,0)

{\fsize


\node[Nmarks=i, Nmr=0](q0)(10,28.5){$q_0$}

\node[Nmarks=n](q1)(30,40){$q_1$}
\node[Nmarks=n](q2)(30,17){$q_2$}

\node[Nmarks=n](q1b)(50,55){$q_1,b$}
\node[Nmarks=n](q1a)(50,40){$q_1,a$}
\node[Nmarks=n](q2b)(50,17){$q_2,b$}
\node[Nmarks=n](q2a)(50,2){$q_2,a$}

\node[Nmarks=n](sink1)(70,55){\frownie}
\node[Nmarks=n](sink2)(70,2){\frownie}

\node[Nmarks=n](q2)(30,17){$q_2$}

\rpnode[Nmarks=n](r1)(70,40)(4,3.5){}
\rpnode[Nmarks=n](r2)(70,17)(4,3.5){}

\node[Nmarks=r](qX)(90,28.5){\smiley}

\drawedge[ELpos=45, ELside=l, ELdist=.5](q0,q1){$-,a$}
\drawedge[ELpos=45, ELside=r, ELdist=.5](q0,q2){$-,b$}


\drawedge[ELpos=50, ELside=l, ELdist=.6, curvedepth=0](q1,q1a){$a,-$}
\drawedge[ELpos=50, ELside=l, ELdist=.6, curvedepth=0](q2,q2b){$b,-$}
\drawedge[ELpos=50, ELside=l, ELdist=0, curvedepth=0](q1,q1b){$b,-$}
\drawedge[ELpos=50, ELside=r, ELdist=0, curvedepth=0](q2,q2a){$a,-$}

\drawedge[ELpos=50, ELside=l, ELdist=.6, curvedepth=0](q1a,r1){$a,-$}
\drawedge[ELpos=50, ELside=l, ELdist=.6, curvedepth=0](q2b,r2){$b,-$}
\drawedge[ELpos=50, ELside=l, ELdist=0, curvedepth=0](q1a,sink1){$b,-$}
\drawedge[ELpos=50, ELside=l, ELdist=0, curvedepth=0](q2b,sink2){$a,-$}

\drawedge[ELpos=50, ELside=r, ELdist=1, curvedepth=-8](q1b,q1){$b,-$}
\drawedge[ELpos=50, ELside=l, ELdist=1, curvedepth=8](q2a,q2){$a,-$}
\drawedge[ELpos=50, ELside=l, ELdist=1, curvedepth=0](q1b,sink1){$a,-$}
\drawedge[ELpos=50, ELside=l, ELdist=1, curvedepth=0](q2a,sink2){$b,-$}

\drawedge[ELpos=18, ELside=l, ELdist=.5, curvedepth=3, sxo=2, syo=1.5, exo=-2](r1,qX){\sfrac{1}{2}}
\drawedge[ELpos=18, ELside=r, ELdist=.5, curvedepth=-3, sxo=2, syo=-1.5, exo=-2](r2,qX){\sfrac{1}{2}}

\drawbpedge[ELpos=31, ELside=r, ELdist=1, eyo=-1](r1,330,25,q1,320,18){\sfrac{1}{2}}
\drawbpedge[ELpos=31, ELside=l, ELdist=1, eyo=1](r2,30,25,q2,40,18){\sfrac{1}{2}}

\drawloop[ELpos=50, ELside=l, ELdist=1, ELside=l,loopCW=y, loopdiam=5, loopangle=90](qX){}
\drawloop[ELpos=50, ELside=l, ELdist=1, ELside=l,loopCW=y, loopdiam=5, loopangle=0](sink1){}
\drawloop[ELpos=50, ELside=l, ELdist=1, ELside=l,loopCW=y, loopdiam=5, loopangle=0](sink2){}



}
\end{picture}
 
\end{center}
\hrule
\caption{{\bf Belief-based strategies are not sufficient.}
The game graph obtained by the reduction of pure to
randomized strategies on the game of \figurename~\ref{fig:GS09-wrong-positive} (for almost-sure reachability objective).
Player~$1$ is blind and player~$2$ has perfect observation. 
There exists an almost-sure winning randomized strategy (with invisible actions),
but there is no \emph{belief-based memoryless} almost-sure winning randomized strategy. 
\label{fig:GS09-wrong-positive-invisible-actions}}

\end{figure}
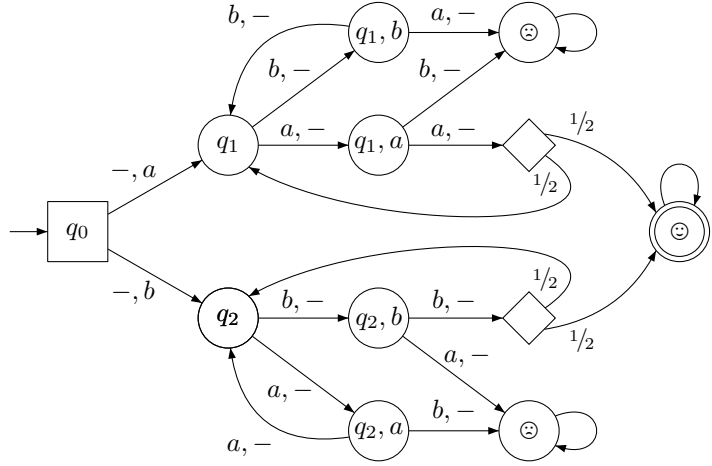

\begin{theorem}
Given a two-sided partial-observation stochastic game 
$G$ with a reachability objective we can construct in time polynomial 
in the size of the game and size of the action sets a
two-sided partial-observation stochastic game $H$
such that there exists a pure almost-sure winning strategy in $G$ iff there exists a 
randomized action-invisible almost-sure winning strategy in $H$.  
\end{theorem}

\smallskip\noindent{\bf Belief-based strategies are not sufficient.}
We illustrate our reduction with the following example that shows
belief-based (belief-only) randomized action-invisible 
strategies are not sufficient for almost-sure reachability in one-sided
partial-observation games (player~1 partial and player~2 perfect), 
showing that a remark (without proof) of~\cite[p.4]{CDHR07} 
and the result and construction of~\cite[Theorem~1]{GS09} are wrong. 

\begin{example}
We illustrate the reduction of 
on the example
of \figurename~\ref{fig:GS09-wrong-positive}. The result of the reduction
is given in \figurename~\ref{fig:GS09-wrong-positive-invisible-actions}.
Remember that Example~\ref{ex:one} showed
that belief-based pure strategies are not sufficient for almost-sure winning.
We show that belief-based randomized strategies are not sufficient for 
almost-sure winning in the game of \figurename~\ref{fig:GS09-wrong-positive-invisible-actions}.
First, in $\{q_1,q_2\}$ player~$1$ has to play pure since he has to be able to 
repeat the same action to avoid reaching a sink state $\frownie$ with positive
probability. Now, the argument is the same as in Example~\ref{ex:one}: playing
always the same action (either $a$ or $b$) in $\{q_1,q_2\}$ is not even positive
winning as player~$2$ can choose the state in this set (either $q_2$ or $q_1$).
\ee
\end{example}

Note that our reduction preserves the structure and memory of almost-sure 
winning strategies,
hence the non-elementary lower bound given in Theorem~\ref{theo:non-elementary}
for pure strategies also transfers to randomized action-invisible strategies 
by the same reduction.

\begin{corollary}\label{coro:non-elementary}
For one-sided partial-observation stochastic games, with player~1 partial and
player~2 perfect, belief-based randomized action-invisible strategies
are not sufficient for almost-sure winning for reachability objectives.
For two-sided partial-observation stochastic games, memory of non-elementary 
size is necessary in general for almost-sure winning for randomized action-invisible 
strategies for reachability objectives.
\end{corollary}

\end{document}